\documentclass[%
reprint,
amsmath,amssymb,
pra,superscriptaddress,
onecolumn
]{revtex4-2}
\usepackage{amsthm}
\usepackage[utf8]{inputenc}
\usepackage{graphicx}
\usepackage{dcolumn}
\usepackage{bm}
\usepackage{mathtools}
\usepackage{xcolor}

\newtheorem{lemma}{Lemma}
\newtheorem{theorem}{Theorem}
\DeclareMathOperator{\Tr}{Tr}
\DeclareMathOperator{\supp}{supp}
\DeclarePairedDelimiterX{\abs}[1]{\vert}{\vert}{#1}
\DeclarePairedDelimiterX{\norm}[1]{\lVert}{\rVert}{#1}
\DeclarePairedDelimiterX{\expval}[1]{\langle}{\rangle}{#1}
\DeclarePairedDelimiterX{\ket}[1]{\vert}{\rangle}{#1}
\DeclarePairedDelimiterX{\bra}[1]{\langle}{\vert}{#1}
\DeclarePairedDelimiterX{\innerproduct}[2]{\langle}{\rangle}{#1\delimsize\vert\mathopen{}#2}
\DeclarePairedDelimiterX{\outerproduct}[2]{\vert}{\vert}{#1\delimsize\rangle\!\delimsize\langle\mathopen{}#2}
\DeclarePairedDelimiterX{\mel}[3]{\langle}{\rangle}%
{#1\delimsize\vert\mathopen{}#2\delimsize\vert\mathopen{}#3}

\makeatletter
\renewenvironment{cases}[1][l]{\matrix@check\cases\env@cases{#1}}{\endarray\right.}
\def\env@cases#1{%
	\let\@ifnextchar\new@ifnextchar
	\left\lbrace\def\arraystretch{1.2}%
	\array{@{}#1@{\quad}l@{}}}
\makeatother

\begin{document}

\title{Perturbation Theory for Quantum Information}

\author{Michael R Grace}
\affiliation{College of Optical Sciences, University of Arizona, Tucson, AZ 85721, USA}
\author{Saikat Guha}
\affiliation{College of Optical Sciences, University of Arizona, Tucson, AZ 85721, USA}

\begin{abstract}
	We report lowest-order series expansions for functions of quantum states based on a perturbation theory for primary matrix functions of linear operators. We show that this Taylor-like representation enables efficient computation of functions of perturbed quantum states that require knowledge only of the eigenspectrum of the unperturbed state and the density matrix elements of a zero-trace, Hermitian perturbation operator, but not requiring analysis of the full perturbed state. We develop this theory for two classes of quantum state perturbations: perturbations that preserve the vector support of the original state and perturbations that extend the support beyond the support of the original state. We highlight relevant features of the two, in particular the fact that functions and measures of perturbed quantum states with preserved support can be elegantly and efficiently represented using Fr\'echet derivatives. We apply our perturbation theory to find simple expressions of Taylor-like expansions for four of the most important quantities in quantum information theory: the von Neumann entropy, the quantum relative entropy, the quantum Chernoff bound, and the quantum fidelity, when their argument density operators are perturbed by small amounts. 
\end{abstract}
\maketitle

\section{Introduction}
\subsection{Motivation}
Many applications of quantum information theory involve analytical evaluation of the effect of small perturbations on various properties of a quantum state. For example, for discriminating among a library of incoherent objects in the sub-diffraction imaging limit, i.e., when the sizes of the objects normalized by the width of the point-spread-function satisfy $\gamma\ll1$, the minimum error probability with $n$ photons is
\begin{equation}
	\begin{split}
		-\frac{1}{n} \log(P_{\rm err}) \sim \,\, &\min_{i,j}\xi\big(\outerproduct{\phi_0}{\phi_0} + \frac{\gamma^2}{2}\rho_i + O(\gamma^3),\\
		&\outerproduct{\phi_0}{\phi_0} + \frac{\gamma^2}{2}\rho_j + O(\gamma^3)\big),
	\end{split}
	\label{eq:example_imaging}
\end{equation}
where $\ket{\phi_0}$ is a known object-independent state, $\rho_i$ and $\rho_j$ depend on the objects in the library, and $\xi(\sigma_1, \sigma_2)$ [Eq.~\eqref{eq:ChernoffInformation}] is the quantum Chernoff exponent~\cite{Grace2021a}. The covert communications capacity of a bosonic channel with QPSK modulation~\cite{Bullock2020} was bounded under a trace-distance covertness constraint, using:
\begin{equation}
	\begin{split}
		1-\sqrt{F(\rho_{\rm QPSK},\rho_{\eta k})}&\leq\norm{\rho_{\rm QPSK},\rho_{\eta k}}_1\\
		&\leq\sqrt{1-F(\rho_{\rm QPSK},\rho_{\eta k})},
	\end{split}
\end{equation}
where $\rho_{\rm QPSK}=\rho_{\eta k}+u^2\tilde{\rho}_{\eta k} + O(u^3)$ and $u$ sets the input power constraint, and $F(\sigma_1, \sigma_2)$~[Eq.~\eqref{eq:Fidelity}] is the Fidelity~\cite{Wang2022}. Relatedly, the covert communications capacity for a general quantum channel (a yet open problem) will derive from an analysis of the Holevo information:
\begin{equation}
	\chi\big(p(x),\rho(x)\big) = S\bigg(\sum_x p(x)\rho(x)\bigg)-\sum_x p(x)S\big(\rho(x)\big),
\end{equation}
for a shrinking ball of states $\rho(x)$ around a known `innocent' state $\rho_0$ with prior probabilities $p(x)$, requiring calculation of the von Neumann entropies $S\big(\rho(x)\big)$~[Eq.~\eqref{eq:VonNeumannEntropy}] for an ensemble of perturbed states~\cite{Bash2015-tv,Gagatsos2020-da}. 
Another place where such perturbation results could be useful is for proving the {\em entropy photon-number inequality} (EPnI)---the quantum version of the entropy-power inequality (EPI)~\cite{Guha2008-fw}---proving which will close the capacity region converse proofs for various multi-user quantum communications settings for bosonic channels~\cite{Guha2007}. A possible proof approach may involve incremental Gaussification of a general quantum state~\cite{Guo2005}, akin to its classical counterpart that leveraged the MMSE-Mutual Information relationship and the EPI~\cite{Verdu2006}. Analysis of perturbed quantum states, such as these examples, form crucial theoretical steps across many other problems within the subfields of quantum computing, communications, sensing and tomography.

\subsection{Measures of quantum states}
Many of the mathematical methods that form the analytical toolbox of quantum physics and quantum information science involve evaluating primary matrix functions of density operators. If $\rho_0=\sum_i \lambda_i\outerproduct{\phi_i}{\phi_i}$ is a density operator describing a quantum state in a Hilbert space $\mathcal{H}$, a primary matrix function $f(x)$ is defined as a map from $\mathcal{H}$ to $\mathcal{H}$ that can be expressed solely in terms of $\mathcal{R}^1$ to $\mathcal{R}^1$ operations on the eigenvalues of the density operator as $f(\rho_0) = \sum_i f(\lambda_i)\outerproduct{\phi_i}{\phi_i}$ \cite{Horn1991,Higham2008}. Evaluating a primary matrix function of a density operator generally requires diagonalization of the state for an exact solution. As such, analytically computing a primary matrix function $f(\rho)$ of an arbitrary state $\rho$ can be less tractable in practice. Our objective in this paper is to provide series expansions for primary matrix functions $f(\rho)$
when the density operator argument has been perturbed to a small degree and, subsequently, to use these to provide simple-to-evaluate Taylor-like expressions for measures of entropy (of a density matrix $\rho$) or distance (between two density matrices $\rho_1$ and $\rho_2$). To these ends, we express perturbed quantum states as $\rho = \rho_0 + \nu$, or $\rho_i = \rho_0 + \nu_i$, $i = 1, 2$, with $\rho_0$ being the {\em unperturbed}, or zeroth-order, state, and $\nu$, $\nu_1$ and $\nu_2$ being zero-trace {\em perturbing operators}, or perturbations, that are characterized by having small Hilbert-Schmidt norms.

Primary matrix functions of density operators appear in a number of useful measures of quantum states that form a broad foundation for computations relevant to applications in quantum state tomography \cite{Christandl2012-cr}, quantum computing \cite{Nielsen2010}, communication of quantum information \cite{Wilde2013}, and quantum-enhanced sensing \cite{Degen2017}. The most fundamental quantity in quantum information theory is the von Neumann entropy, which is defined as \cite{Wilde2013}
\begin{equation}
	S(\rho) = -\Tr\big[\rho\log(\rho)\big].
	\label{eq:VonNeumannEntropy}
\end{equation}
The operational interpretation of $S(\rho)$, in analogy to the classical Shannon entropy $H(X)$ of a random variable $X$, is in the compressibility of quantum information, viz., $n$ copies of $\rho$ can be unitarily encoded (compressed) into a quantum register of $n S(\rho)$ qubits, and can be losslessly decoded (uncompressed), in the limit that $n \to \infty$. The von Neumann entropy also forms the basis of many other quantum information quantities relating to quantum channels, such as the {\em Holevo information} that quantifies the classical communication capacity of a quantum channel~\cite{holevo1973bounds}, the {\em quantum mutual information} that quantifies the entanglement-assisted classical communication capacity of a channel~\cite{bennett1999entanglement}, and the {\em coherent information} that quantifies the quantum communication capacity of a channel~\cite{ieee2005dev}. 

Several measures quantify the distinguishability or similarity between two quantum states $\rho_1$ and $\rho_2$ in the same Hilbert space. One of these is the quantum relative entropy (QRE):
\begin{equation}
	D(\rho_1||\rho_2) = \Tr\big[\rho_1\big(\log(\rho_1)-\log(\rho_2)\big)\big],
	\label{eq:RelativeEntropy}
\end{equation}
an asymmetric measure that extends the  Kullback–Leibler (KL) divergence between two probability distributions $D(P_X||P_Y)$ of random variables $X$ and $Y$. The QRE can be used to express the quantum mutual information between two quantum systems $A$ and $B$, viz., $I(A;B) = D(\rho_{AB}||\rho_{A} \otimes \rho_B)$. It also has applications in bounding the probability of error in distinguishing between two quantum states $\rho_1$ and $\rho_2$ through Pinsker's inequality~\cite{Wilde2013}. Next, the quantity:
\begin{equation}
	\xi_s(\rho_1,\rho_2) = -\log\big(\Tr[\rho_1^s\rho_2^{1-s}]\big),
	\label{eq:ChernoffInformation}
\end{equation}
is a symmetric distance measure between two states, which maximized over $s \in [0, 1]$ results in the quantum Chernoff bound (QCB), $\xi(\rho_1,\rho_2)=\max_{s\in[0,1]}\xi_s(\rho_1,\rho_2)$. The QCB sets an asymptotically tight upper bound on the asymptotic error-probability exponent of the multi-copy binary hypothesis test $\rho_1^{\otimes n}$ versus $\rho_2^{\otimes n}$~\cite{Audenaert2007}. Finally, the quantum fidelity:
\begin{equation}
	F(\rho_1,\rho_2) = \Tr\bigg[\sqrt{\sqrt{\rho_1}\rho_2\sqrt{\rho_1}}\bigg]^2,
	\label{eq:Fidelity}
\end{equation}
is another symmetric distance measure between $\rho_1$ and $\rho_2$ that is commonly used to quantify the similarity of two states. It also appears in the definitions of a number of information theoretic bounds such as the quantum Cram\'er-Rao bound, through the Bures distance \cite{Braunstein1994}. All of the abovesaid four quantities depend on primary matrix functions (e.g., $\log(x)$ and $x^s$) of density operators that are evaluated by diagonalizing the state(s) and operating on the eigenvalues.

\subsection{Summary of main results}

\begin{table}[tbp]
	\centering
	\begin{tabular}{|p{0.125\linewidth} | p{0.41\linewidth}| p{0.465\linewidth}|}
		\hline
		
		& Support-Preserving Perturbation(s) 
		& Support-Extending Perturbation(s) 
		\\
		\hline \hline
		Von Neumann Entropy & \vspace{-2pt} $S(\rho)\approx S(\rho_0)-\Tr\big[\nu\log(\rho_0)\big]-\frac{1}{2}\Tr\big[\nu L_{\log(x)}(\rho_0,\nu)\big]$ & 
		\vspace{-2pt}$S(\rho)\approx S(\rho_0)-\Tr[L_{(x)\log(x)}(\rho_0,\nu_{\rm B})]-\Tr[\nu_{\rm D}\log(\nu_{\rm D})]$\\
		\hline
		Quantum Relative Entropy & \vspace{-2pt}$D(\rho_1||\rho_2)\approx \frac{1}{2}\Tr[(\nu_1-\nu_2)L_{\log(x)}(\rho_0,\nu_1-\nu_2)]$ & 
		\vspace{-2pt}$D(\rho_1||\rho_2)\approx \Tr[\nu_{1, \rm B}-\nu_{2,\rm B}]+\Tr\big[\nu_{1,\rm D}\big(\log(\nu_{1,\rm D})-\log(\nu_{2,\rm D})\big)\big]$\\
		\hline
		Quantum Chernoff Bound & \vspace{-2pt}$\xi(\rho_1,\rho_2)\approx \frac{1}{2}\Tr\big[L_{\sqrt{x}}(\rho_0,\nu_1-\nu_2)^2\big]$&
		\vspace{-2pt}$\xi_s(\rho_1,\rho_2)\approx -s\Tr[\nu_{1,\rm B}]-(1-s)\Tr[\nu_{2,\rm B}]-\Tr[\nu_{1,\rm D}^s\nu_{2,\rm D}^{1-s}]$\\
		\hline
		Quantum Fidelity & \vspace{-2pt}$F(\rho_1,\rho_2)\approx 1-\frac{1}{2}\Tr\big[(\nu_1-\nu_2)L_{\sqrt{x}}(\rho_0^2,\nu_1-\nu_2)\big]$&
		\vspace{-2pt}$F(\rho_1,\rho_2)\approx 1+\Tr[\nu_{1,\rm B}+\nu_{2,\rm B}]+2\Tr\big[\sqrt{\sqrt{\nu_{1,\rm D}}\nu_{2,\rm D}\sqrt{\nu_{1,\rm D}}}\big]$\\
		\hline
	\end{tabular}
	\caption{Taylor-like expansions of common measures for quantum information theory for perturbed quantum states. The perturbations are taken as: $\rho = \rho_0 + \nu$, and $\rho_i = \rho_0 + \nu_i$, $i = 1, 2$, with $\rho_0$ being the {\em unperturbed} or zeroth-order state, and $\nu$, $\nu_1$ and $\nu_2$ are zero-trace perturbing operators of small Hilbert-Schmidt norms. Second-order approximations in the Hilbert-Schmidt norm of the perturbing operator(s) are given for support-preserving perturbations, while first-order approximations are given for support-extending perturbations. All definitions as well as the explicit scaling of the remainder terms with respect to the Hilbert-Schmidt norms of the perturbing operators are given in Sections~\ref{sec:Model}, \ref{sec:Results_Support-Preserving} and \ref{sec:Results_Support-Extending}.}
	\label{tab:MainResults}
\end{table}

In this work, we develop a perturbation theory for quantum information around the insight that, as in scalar calculus, approximations for primary matrix functions of perturbed quantum states can be obtained through lowest-order series expansions. The value of perturbation theories in applied mathematics is the ability to infer low-order properties of a quantity (in terms of power series with respect to a measure of the size of the perturbation) from a well-characterized zeroth-order quantity without having to further analyze the perturbation. A primary matrix function $f(\rho)$ acting on a density operator $\rho$ that deviates from a diagonalizable density operator $\rho_0$ by a small perturbation operator $\nu$ is a perfect candidate for the use of perturbation theory. 

Under a set of conditions given in the next sections, the main results of our work are (1) an application of the operator perturbation theory of Dalecki\u{\i} and Kre\u{\i}n to find Taylor-like series expansions for primary matrix functions of perturbed quantum states (see Theorems~\ref{thm:QuantumDaleckii-Krein2} and \ref{thm:QuantumDaleckii-Krein1SupportExtending}) and (2) analytical lowest-order series expansions for the measures listed in Eqs.~\ref{eq:VonNeumannEntropy}-\ref{eq:Fidelity}. In developing each of these results, we define two different classes of perturbation operators ($\nu$, $\nu_1$, and $\nu_2$) on $\mathcal{H}$: ``support-preserving" perturbations, whose Hilbert space supports are fully contained within the support of the zeroth-order quantum state $\rho_0$, and ``support-extending" perturbations, whose supports extend beyond the support of the zeroth-order state. In the specific case of support-preserving perturbations, we show that the second order terms in the series expansions for each of the four measures in Eqs.~\ref{eq:VonNeumannEntropy}-\ref{eq:Fidelity} take on elegant forms that depend on Fr\'echet derivatives $L_{f(x)}(\rho_0,\nu)$ of primary matrix functions $f(x)$ evaluated on perturbed density matrices $\rho=\rho_0+\nu$. The use of Fr\'echet derivatives unifies our support-preserving perturbation theory, which enables analytic evaluation of quantum information theoretic measures of perturbed states while only requiring eigenanalysis of the unperturbed state. We concisely summarize our results in Table~\ref{tab:MainResults} for ease of reference, where definitions of quantum states and operator derivatives are given in the following sections. From Table~\ref{tab:MainResults} we can draw the following key observations: 
\begin{itemize}
	\item All four second-order expansions in the case of support-preserving perturbations depend only on first-order (i.e., Fr\'echet) derivatives of primary matrix functions; furthermore, neither the QRE, QCB, nor fidelity have any first-order contribution with respect to $\norm{\nu_1}$ or $\norm{\nu_2}$, and the only dependence of these quantities on the perturbations appears in the difference $\nu_1-\nu_2$ (see Theorems~\ref{thm:RelativeEntropySupportPreserving}-\ref{thm:FidelitySupportPreserving}).
	\item Our second-order expansion of the QCB for support-preserving perturbations is attained with $s=1/2$, and we therefore confirm that the QCB converges to the quantum Bhattacharyya bound \cite{Pirandola2008} for two states separated by small perturbations (see Theorem~\ref{thm:ChernoffInfoSupportPreserving}).
	\item Our second-order expression for the quantum fidelity can be directly used to derive the exact analytical form for the Bures distance \cite{Hubner1992} and therefore the quantum Fisher information \cite{Braunstein1994} when the support of the state is preserved (see Theorem~\ref{thm:FidelitySupportPreserving}); on the other hand, our first-order expression with support-extending perturbations reflects the discontinuity inherent to the Bures distance \cite{Safranek2017,Zhou2019c} when the rank of the state changes (see Theorem~\ref{thm:FidelitySupportExtending}).
	\item Our first-order expansions for support-extending perturbations reveal that quantum information theoretic measures of distance between perturbed states have no explicit dependence on $\rho_0$ but can be computed as sums between, firstly, partial traces of the perturbations over the support of $\rho_0$ and, secondly, the distance measure evaluated between the two perturbations over the kernel of $\rho_0$ (see Theorems~\ref{thm:RelativeEntropySupportExtending}-\ref{thm:FidelitySupportExtending}).
\end{itemize}
We elaborate further on these observations in the following sections. Our formalism has already been utilized to perform crucial analytical steps in aforementioned analyses of the quantum limits of sub-diffraction imaging~\cite{Grace2022} and covert communications~\cite{Wang2022} and will be useful for many other applications.

The paper is organized as follows. We first in Section \ref{sec:Model} define a quantum model for perturbations of a general mixed quantum state on a Hilbert space. We also provide background on the theory of matrix perturbations, including the concept of the Fr\'echet derivative. In Section \ref{sec:Results_Support-Preserving} we provide an expression for the second-order behavior of a primary matrix function of a quantum state perturbed by a small, support-preserving linear operator with respect to the Hilbert-Schmidt norm of the perturbing operator. We use this perturbation theory to find series expansions for several commonly occurring measures in quantum information theory for support-preserving perturbations, which are easily evaluated to second-order using Fr\'echet derivatives. In Section \ref{sec:Results_Support-Extending} we derive a corresponding first-order series expansion for the particular primary matrix function $f(x)=x^s$, $s\in(0,1)$, and use it to find series expansions for the same quantum information theoretic measures for support-extending perturbations. We provide closing remarks in Section \ref{sec:Discussion}.

\section{Model for Quantum State Perturbations}
\label{sec:Model}
Consider a general quantum state described by its density operator $\rho_0$ on a Hilbert space $\mathcal{H}$ with spectral decomposition $\rho_0=\sum_i \lambda_i \outerproduct{\phi_i}{\phi_i}$, with normalization condition $\Tr(\rho_0)=1$. We introduce a quantum state perturbation $\nu$ as an operator on $\mathcal{H}$, such that the resulting perturbed state is $\rho=\rho_0+\nu$. The following three properties must be satisfied for $\nu$ to be a valid perturbation for the state $\rho_0$. First, $\nu$ must be Hermitian, so that $\rho\in\mathcal{H}$. Second, $\Tr[\nu]=0$, so that $\rho$ is a properly normalized, unit-trace quantum state. Third, for all $i\in [1, \dim(\mathcal{H})]$,  $-\lambda_i\leq\mel{\phi_i}{\nu}{\phi_i}\leq1-\lambda_i$, so that $\rho$ remains positive semi-definite. These properties hold for finite- or infinite-dimensional Hilbert spaces; our results as proven here are only formulated for finite-dimensional state spaces (or those that can be truncated to finite dimensions), but the results are likely extendable to infinite-dimensional Hilbert spaces. For the purposes of constructing a perturbation theory for quantum states, we will consider ``small" perturbations with a Hilbert-Schmidt norm $\norm{\nu}$ that can be quantified by the condition $\norm{\nu}\ll\epsilon$ for some small real number $\epsilon$. When $\mathcal{H}$ is finite-dimensional, the Hilbert-Schmidt norm is equivalent to the Frobenius norm.

\subsection{Operator derivatives and small-perturbation expansions}
The calculus of functions of matrices has a rich mathematical history \cite{Daleckii1965,Bhatia1996,Higham2008,higham2014,Mathias1996,Peller2006,DelMoral2018}. In particular, the differential effect of a function $f(x)$ at a Hermitian matrix $A$ with respect to another Hermitian matrix $E$ is captured by the Fr\'echet derivative
\begin{equation}
	L_{f(x)}(A,E) = \frac{d}{dt}\Big\vert_{t=0}f(A+t E),
	\label{eq:FrechetDerivative}
\end{equation}
which exists if Eq.~\ref{eq:FrechetDerivative} produces a matrix $L_f(A,E)$ such that $\norm{f(A+E)-f(A)-L_{f(x)}(A,E)}=o(\norm{E})$ \cite{Higham2008,Bhatia1996}. We denote real valued derivatives of the function $f(x)$ as $f'(x)=\partial f(x)/\partial x$, $f''(x)=\partial^2 f(x)/\partial x^2$, etc. When working in the eigenbasis of $A=U\Lambda_{\vec{\alpha}}U^{\dagger}$, where $\Lambda_{\vec{\alpha}}$ is a diagonal matrix containing the vector of eigenvalues $\vec{\alpha}$ of $A$, the Fr\'echet derivative takes the computable form \cite{Daleckii1951,Bhatia1996,Higham2008,Carlsson2019a}
\begin{equation}
	L_{f(x)}(A,E) = U\big([f(x),\vec{\alpha}]^{[1]}\circ\hat{E}\big)U^{\dagger},
	\label{eq:FrechetDerivative2}
\end{equation}
where the symmetric matrix $[f(x),\vec{\alpha}]^{[1]}$ is the first divided difference of the function $f(x)$ at $A$, defined by
\begin{equation}
	[f(x),\vec{\alpha}]^{[1]}_{k,l}=
	\begin{cases}[c]
		\big(f(\alpha_k)-f(\alpha_l)\big)/\big(\alpha_k-\alpha_l\big) \quad &\alpha_k\neq\alpha_l\\
		f'(\alpha_k)\quad &\alpha_k=\alpha_l,
	\end{cases}
	\label{eq:DividedDifference1}
\end{equation}
$\hat{E}=U^{\dagger}EU$, and $\circ$ denotes the Hadamard (element-wise) product. In the context of quantum state perturbations, the Fr\'echet derivative $L_{f(x)}(\rho_0,\nu)$ can be understood as the derivative of the function $f(x)$ at a state $\rho_0$ in the \emph{direction} of the perturbation $\nu$. The definition of the Fr\'echet derivative in Eq.~\ref{eq:FrechetDerivative2} requires only a spectral decomposition of $\rho_0$ and not of $\nu$. 

Matrix differentiation can be used to find series expansions for functions of perturbed matrices. Our results depend on a theorem proved by Dalecki\u{\i} and Kre\u{\i}n \cite{Daleckii1965} that states that for a full rank Hermitian matrix $A$ on a Hilbert space $\mathcal{H}$ perturbed linearly by another Hermitian matrix $E$, a Taylor-like expansion about $\epsilon=0$ is given by
\begin{equation}
	\begin{split}
		f(A+\epsilon E) =& f(A) + \sum_{k=1}^{K}\frac{\epsilon^k}{k!}D^{[k]}_{f(x)}(A,E)+R_K(\epsilon),
	\end{split}
	\label{eq:Daleckii-Krein-Full}
\end{equation}
where 
\begin{equation}
	D^{[k]}_{f(x)}(A,E)=\frac{d^k}{dt^k}\Big\vert_{t=0}f(A+t E)
	\label{eq:MatrixDerivatives}
\end{equation}
and where $R_K(\epsilon)$ is a residual term that depends linearly on $\big(\epsilon\norm{E}\big)^{K+1}$ \cite{Daleckii1965}. Clearly, $D^{[1]}_{f(x)}(A,E)=L_{f(x)}(A,E)$. While at face value Eq.~\eqref{eq:Daleckii-Krein-Full} is only useful for convex sums of two matrices governed by the small scalar perturbation factor $\epsilon\ll1$, setting $\epsilon=1$ allows for series expansions of a matrix $A$ perturbed by a second matrix $E$, with the residual depending on powers of $\norm{E}$ \cite{Carlsson2019}. The trivial zeroth-order expansion is
\begin{equation}
	f(A+E)=f(A) + O(\norm{E}).
	\label{eq:Daleckii-Krein0}
\end{equation}
We refer to Eq.~\ref{eq:Daleckii-Krein0} as the zeroth-order Dalecki\u{\i}-Kre\u{\i}n expansion. A first-order Dalecki\u{\i}-Kre\u{\i}n expansion makes use of the Fr\'echet derivative \cite{Carlsson2019}:
\begin{equation}
	f(A+E)=f(A) + U\big([f(x),\vec{\alpha}]^{[1]}\circ\hat{E}\big)U^{\dagger} + O(\norm{E}^2).
	\label{eq:Daleckii-Krein1}
\end{equation}
For higher-order terms in the expansion, we need the higher-order matrix derivatives found in Eq.~\ref{eq:MatrixDerivatives}. These are not equivalent to higher-order Fr\'echet derivatives \cite{higham2014} but can be still be computed through established tensor calculus tools \cite{Mathias1996,Peller2006,Sendov2007}. The second derivative is computed by \cite{Daleckii1965,Bhatia1996}
\begin{equation}
	D^{[2]}_{f(x)}(A,E)=2\sum_{k,l,m}[f(x),\vec{\alpha}]^{[2]}_{k,l,m}U_k \hat{E} U_l \hat{E} U_m,
	\label{eq:2ndDerivative}
\end{equation}
where the $U_k$ are projections onto the eigenvectors of $A$ and the second divided difference tensor $[f(x),\vec{\alpha}]^{[2]}$ is a symmetric tensor defined by
\begin{equation}
	[f(x),\vec{\alpha}]^{[2]}_{k,l,m}=
	\begin{cases}[c]
		\big([f(x),\vec{\alpha}]^{[1]}_{k,l}-[f(x),\vec{\alpha}]^{[1]}_{m,l}\big)/(\alpha_k-\alpha_m)\quad& \alpha_k\neq\alpha_m\\
		\big(f'(\alpha_k)-[f(x),\vec{\alpha}]^{[1]}_{k,l}\big)/(\alpha_k-\alpha_l)\quad& \alpha_k=\alpha_m\neq\alpha_l\\
		\frac{1}{2}f''(\alpha_k)\quad&\alpha_k=\alpha_m=\alpha_l.
	\end{cases}
	\label{eq:DividedDifference2}
\end{equation}
The second-order Dalecki\u{\i}-Kre\u{\i}n expansion is then given by
\begin{equation}
	f(A+E)=f(A) + U\big([f(x),\vec{\alpha}]^{[1]}\circ\hat{E}\big)U^{\dagger}+U\bigg(\sum_{k,l,m}[f(x),\vec{\alpha}]^{[2]}_{k,l,m}U_k \hat{E} U_l \hat{E} U_m\bigg)U^{\dagger}+O(\norm{E}^3).
	\label{eq:Daleckii-Krein2}
\end{equation}

The mathematical literature for perturbation theory for primary matrix functions of singular matrices is more limited, as eigenvalues equaling zero severely restrict Fr\'echet differentiability, but we make use of recent results to apply matrix perturbation theory to quantum states with support-extending perturbations \cite{Carlsson2018}. Consider the Hilbert space decomposition $\mathcal{H}=\mathcal{H}_+\oplus\mathcal{H}_0$, where $\mathcal{H}_+=\supp(\rho_0)$ and $\mathcal{H}_0=\ker(\rho_0)$ are subspaces corresponding to the support and the kernel of $A$, respectively \cite{Wilde2013}. We will use a Taylor-like expansion of $(A+E)^{1/p}=U(\Lambda_{\vec{\alpha}}+\hat{E})^{1/p}U^{\dagger}$ for small perturbing matrices when the vector support of $E$ extends to $\mathcal{H}_0$. We will use the block decompositions
\begin{eqnarray}
	\label{eq:BlockDecompositionA} A&=&U\Lambda_{\vec{\alpha}}U^{\dagger}=U
	\begin{pmatrix}
		\Lambda_{\vec{\alpha}_+} & 0\\
		0 & 0
	\end{pmatrix}
	U^{\dagger}\\
	\label{eq:BlockDecompositionE} E&=&U\hat{E}U^{\dagger}=U
	\begin{pmatrix}
		B & C\\
		C^{\dagger} & D
	\end{pmatrix}
	U^{\dagger}
\end{eqnarray}
where $\Lambda_{\vec{\alpha}_+}$ is a diagonal matrix whose diagonal elements correspond to $\vec{\alpha}_+$, the vector of nonzero eigenvalues of $A$. Define $\bar{D}=D-C^{\dagger}(\Lambda_{\vec{\alpha}_+}+B)^{-1}C$ to be the Schur complement of the upper left block of $\Lambda_{\vec{\alpha}}+\hat{E}$. The first divided difference of $f(x) = x^{q}$ is extended into $\mathcal{H}_0$ as
\begin{equation}
	\big[x^q,\vec{\alpha}\big]^{[1,0]}_{k,l}=
	\begin{cases}[c]
		\big(\alpha_k^{q}-\alpha_l^{q}\big)\big(\alpha_k-\alpha_l\big) \quad &\alpha_k\neq\alpha_l\\
		q\alpha_k^{q-1}\quad &\alpha_k=\alpha_l>0\\
		1 \quad &\alpha_k=\alpha_l=0.
	\end{cases}
	\label{eq:DividedDifferenceCarlsson}
\end{equation}
For the $p^{th}$ root function $f(x)=x^{1/p}$, $1<p<\infty$, a Daleckii-Kre\u{\i}n-like expansion was proven for a singular Hermitian matrix $A$ perturbed by a second Hermitian matrix $E$, giving \cite{Carlsson2018}  
\begin{equation}
	(A+E)^{1/p} = A^{1/p} + U\bigg(\big[x^{1/p},\vec{\alpha}\big]^{[1,0]}\circ
	\begin{pmatrix}
		B & C\\
		C^{\dagger} & \bar{D}^{1/p}
	\end{pmatrix}
	\bigg)U^{\dagger}+ O(\norm{E}^r)
	\label{eq:Daleckii-Krein_Singular}
\end{equation}
where $r = \min(1+1/p,3/p)$. When $1<p<3$, Eq.~\ref{eq:Daleckii-Krein_Singular} is a useful first-order expansion for $(A+E)^{1/p}$, while for $p\geq3$ the expression is still valid but does not provide adequate control on the residual term to serve as a lowest-order expansion.

\section{Support-Preserving Perturbation Theory}
\label{sec:Results_Support-Preserving}

We first consider perturbations whose support is spanned by the zeroth-order state, i.e., $\supp(\nu)\subseteq\supp(\rho_0)$, such that the support of the resulting state $\rho$ is not extended beyond that of $\rho_0$. The most obvious sufficient condition thereof is when $\rho_0$ is full rank on $\mathcal{H}$, but this is not necessary; it is possible that $\supp(\nu)\subseteq\supp(\rho_0)$ is satisfied if both $\rho_0$ and $\nu$ are rank-deficient.

Our first main result is a direct application of the second-order Dalecki\u{\i}-Kre\u{\i}n expansion (Eq.~\ref{eq:Daleckii-Krein2}) to a primary matrix function  of a perturbed quantum state.
\begin{theorem}
	Let $\rho_0=\sum_i \lambda_i \outerproduct{\phi_i}{\phi_i}$ be a unit-trace quantum state on $\mathcal{H}$, and let $\nu$ be an operator on $\mathcal{H}$ whose support is a subspace of the support of $\rho_0$. If $f(x)$ is a primary matrix function that is $C^6$ at the eigenvalues $\vec{\lambda}$ such that $f(\rho_0)=\sum_i f(\lambda_i) \outerproduct{\phi_i}{\phi_i}$, then a second-order series expansion in $\norm{\nu}$ for the function acting on the state $\rho = \rho_0+\nu$ is given by
	\begin{equation}
		f(\rho)=f(\rho_0)+L_{f(x)}(\rho_0,\nu)+\frac{1}{2}D_{f(x)}^{[2]}(\rho_0,\nu)+O(\norm{\nu}^3).
		\label{eq:QuantumDaleckii-Krein2}
	\end{equation}
	\label{thm:QuantumDaleckii-Krein2}
\end{theorem}
\begin{proof}
	Let $A=\sum_k \lambda_k \outerproduct{\phi_k}{\phi_k}$ and $E=\sum_{k,l} \mel{\phi_k}{\nu}{\phi_l}\outerproduct{\phi_k}{\phi_l}$ be the $\rho_0$-eigenbasis representations of the density operators $\rho_0$ and $\nu$, respectively,  $\vec{\alpha}=\vec{\lambda}$, and $U_k=\outerproduct{\phi_k}{\phi_k}$ via the quantum mechanical description of the linear projector operator. From Eqs.~\ref{eq:FrechetDerivative2} and \ref{eq:2ndDerivative} we have
	\begin{eqnarray}
		\label{eq:FrechetDerivativeQuantum}L_{f(x)}(\rho_0,\nu)&=&\sum_{k,l} [f(x),\vec{\lambda}]^{[1]}_{k,l} \mel{\phi_k}{\nu}{\phi_l}\outerproduct{\phi_k}{\phi_l}\\
		\label{eq:2ndDerivativeQuantum}D_{f(x)}^{[2]}(\rho_0,\nu)&=&2\sum_{k,l,m} [f(x),\vec{\lambda}]_{k,l,m}^{[2]}\mel{\phi_k}{\nu}{\phi_l}\mel{\phi_l}{\nu}{\phi_m}\outerproduct{\phi_k}{\phi_m},
	\end{eqnarray}
	and the basis-independent Eq.~\ref{eq:QuantumDaleckii-Krein2} directly follows from Eq.~\ref{eq:Daleckii-Krein2}.
\end{proof}
The result in Theorem~\ref{thm:QuantumDaleckii-Krein2} has several desirable features for quantum information theory. First, and of greatest practical significance, this state perturbation theory only requires diagonalization of the zeroth-order state $\rho_0$, whereas the perturbing term $\nu$ only contributes through simply reading off density matrix elements $\mel{\phi_l}{\nu}{\phi_m}$ in the eigenbasis of $\rho_0$. Second, the expression holds for any primary matrix function $f(x)$ as long as $f(\lambda)$ is sufficiently differentiable at each of the eigenvalues $\vec{\lambda}$. Third, the expression can be utilized for any perturbed quantum state $\rho=\rho_0+\nu$ and is not restricted to convex sums $\rho=\rho_0+\epsilon\nu$ under a small linear parameter $\epsilon$. The latter can arise when the state is prepared via evolution under a perturbed Hamiltonian, e.g., when using Lie-Trotter-Suzuki approximations of exponential operators for simulation of quantum systems~\cite{Suzuki1985}, or when a state is subjected to a perturbed quantum channel. On the other hand, Eq.~\ref{eq:QuantumDaleckii-Krein2} also applies for states $\rho$ that are mathematically ``close" to a well characterized state $\rho_0$ (quantified by $\norm{\nu}\ll1$) even if no physical perturbing process can be identified that maps the two states to one another, an example being various proposals for the preparation of approximate ``cat" states in continuous-variable quantum information~\cite{Ourjoumtsev2007,Etesse2015}.

In the remainder of this section we use Theorem~\ref{thm:QuantumDaleckii-Krein2} to find second-order expansions of the quantum information theoretic quantities given in Table~\ref{tab:MainResults}. In each case, analytical evaluation on perturbed quantum states in general requires diagonalization of the full state $\rho$. By utilizing the following second-order expansions, only the zeroth-order states need be diagonalized, while the matrix elements of the perturbation operator can simply be read off in the eigenbasis of $\rho_0$. We begin with entropic quantities, for which we use the following two lemmas for the trace of matrix derivatives.

\begin{lemma}
	Let $\rho_0=\sum_i \lambda_i \outerproduct{\phi_i}{\phi_i}$ and $\nu$ be operators on $\mathcal{H}$, where the support of $\nu$ is a subspace of the support of $\rho_0$. If $f(x)$ is a primary matrix function that is $C^1$ at the eigenvalues $\vec{\lambda}$, 
	\begin{equation}
		\Tr\Big[L_{f(x)}(\rho_0,\nu)\Big]=\Tr[\nu f'(\rho_0)].
		\label{eq:Frechet1LemmaEntropies}
	\end{equation}
	\label{lem:Frechet1LemmaEntropies}
\end{lemma}

\begin{proof}
	Eq.~\ref{eq:Frechet1LemmaEntropies} follows from the quantum-mechanical definition of the Fr\'echet derivative (Eq.~\ref{eq:FrechetDerivativeQuantum}):
	\begin{equation*}
		\begin{split}
			\Tr\big[L_{f(x)}(\rho_0,\nu)\big]=&\Tr\Big[\sum_{k,l}[f(x),\vec{\lambda}]^{[1]}_{k,l}\mel{\phi_k}{\nu}{\phi_l}\outerproduct{\phi_k}{\phi_l}\Big]\\
			=&\Tr\Big[\sum_{k}f'(\lambda_k)\mel{\phi_k}{\nu}{\phi_k}\outerproduct{\phi_k}{\phi_k}\Big]\\
			=&\Tr\Big[\sum_{k}\mel{\phi_k}{\nu}{\phi_k}\outerproduct{\phi_k}{\phi_k}\sum_{i}f'(\lambda_i)\outerproduct{\phi_i}{\phi_i}\Big]\\
			=&\Tr[\nu f'(\rho_0)].
		\end{split}
	\end{equation*}
\end{proof}

\begin{lemma}
	Let $\rho_0=\sum_i \lambda_i \outerproduct{\phi_i}{\phi_i}$ and $\nu$ be operators on $\mathcal{H}$, where the support of $\nu$ is a subspace of the support of $\rho_0$. If $f(x)$ is a primary matrix function that is $C^2$ at the eigenvalues $\vec{\lambda}$, 
	\begin{equation}
		\Tr\Big[D_{f(x)}^{[2]}(\rho_0,\nu)\Big]=\Tr\Big[\nu L_{f'(x)}(\rho_0,\nu)\Big].
		\label{eq:Frechet2LemmaEntropies}
	\end{equation}
	\label{lem:Frechet2LemmaEntropies}
\end{lemma}

\begin{proof}
	Applying Eq.~\ref{eq:2ndDerivativeQuantum}, we have
	\begin{equation*}
		\begin{split}
			\Tr\big[D_{f(x)}^{[2]}(\rho_0,\nu)\Big]=&2\Tr\Big[\sum_{k,l,m} [f(x),\vec{\lambda}]_{k,l,m}^{[2]}\mel{\phi_k}{\nu}{\phi_l}\mel{\phi_l}{\nu}{\phi_m}\outerproduct{\phi_k}{\phi_m}\Big]\\
			=&2\sum_{k,l} [f(x),\vec{\lambda}]_{k,l,k}^{[2]}\abs{\mel{\phi_k}{\nu}{\phi_l}}^2\\
			=&2\bigg(\sum_{k<l}\Big([f(x),\vec{\lambda}]_{k,l,k}^{[2]}+[f(x),\vec{\lambda}]_{l,k,l}^{[2]}\Big)\abs{\mel{\phi_k}{\nu}{\phi_l}}^2+\sum_{k}[f(x),\vec{\lambda}]_{k,k,k}^{[2]}\abs{\mel{\phi_k}{\nu}{\phi_k}}^2\bigg),
		\end{split}
	\end{equation*}
	where the final line splits the sum apart and rearranges indices. Recognizing that $[f(x),\vec{\lambda}]^{[1]}_{k,l}=[f(x),\vec{\lambda}]^{[1]}_{l,k}$, we get a cancellation of first divided differences in the first sum and find
	\begin{equation*}
		\begin{split}
			\Tr\big[D_{f(x)}^{[2]}(\rho_0,\nu)\Big]=&2\bigg(\sum_{k<l}\Big(\frac{f'(\lambda_k)}{\lambda_k-\lambda_l}+\frac{f'(\lambda_l)}{\lambda_l-\lambda_k}\Big)\abs{\mel{\phi_k}{\nu}{\phi_l}}^2+\frac{1}{2}\sum_{k}f''(\lambda_k)\abs{\mel{\phi_k}{\nu}{\phi_k}}^2\bigg)\\
			=&2\sum_{k<l}[f'(x),\vec{\lambda}]^{[1]}_{k,l}\abs{\mel{\phi_k}{\nu}{\phi_l}}^2+\sum_{k}[f'(x),\vec{\lambda}]^{[1]}_{k,k}\abs{\mel{\phi_k}{\nu}{\phi_k}}^2\\
			=&\sum_{k,l} [f'(x),\vec{\lambda}]_{k,l}^{1]}\abs{\mel{\phi_k}{\nu}{\phi_l}}^2\\
			=&\Tr\Big[\sum_{i,j}\mel{\phi_i}{\nu}{\phi_j}\outerproduct{\phi_i}{\phi_j}\sum_{k,l} [f'(x),\vec{\lambda}]_{k,l}^{[1]}\mel{\phi_k}{\nu}{\phi_l}\outerproduct{\phi_k}{\phi_l}\Big]\\
			=&\Tr\Big[\nu L_{f'(x)}(\rho_0,\nu)\Big],
		\end{split}
	\end{equation*}
	where the last equality uses the quantum-mechanical definition of the Fr\'echet derivative (Eq.~\ref{eq:FrechetDerivativeQuantum}).
\end{proof}

First, we find a second-order expansion for the von Neumann entropy. 

\begin{theorem}
	For a perturbed quantum state $\rho=\rho_0+\nu$ on $\mathcal{H}$ where $\rho_0=\sum_i \lambda_i \outerproduct{\phi_i}{\phi_i}$ is a density operator on $\mathcal{H}$ and $\nu$ is a zero-trace state perturbation on $\mathcal{H}$ whose support is a subspace of the support of $\rho_0$, the von Neumann entropy of $\rho$ is given by
	\begin{equation}
		S(\rho)=S(\rho_0)-\Tr\big[\nu\log(\rho_0)\big]-\frac{1}{2}\Tr\big[\nu L_{\log(x)}(\rho_0,\nu)\big]+O(\norm{\nu}^3).
		\label{eq:VonNeumannEntropySupportPreserving}
	\end{equation}
	\label{thm:VonNeumannEntropySupportPreserving}
\end{theorem}

\begin{proof}
	We start by applying Theorem~\ref{thm:QuantumDaleckii-Krein2} with $f(x)=x\log(x)$ to expand the argument of the trace in the definition of the von Neumann entropy given in Eq.~\ref{eq:VonNeumannEntropy}. The zeroth-order term in $\norm{\nu}$ in Eq. \ref{eq:VonNeumannEntropySupportPreserving} is trivial. For the first-order term, we use Lemma~\ref{lem:Frechet1LemmaEntropies} to evaluate the Fr\'echet derivative
	\begin{equation*}
		\begin{split}
			\Tr\big[L_{x\log(x)}(\rho_0,\nu)\big]
			=&\Tr\big[\nu\big(\mathcal{I}+\log(\rho_0)\big)\big]\\
			=&\Tr\big[\nu\log(\rho_0)\big],
		\end{split}
	\end{equation*}
	where we have made use of the property $\Tr[\nu]=0$ for matrix perturbations. Using Lemma~\ref{lem:Frechet2LemmaEntropies}, the second-order term in $\norm{\nu}$ becomes
	\begin{equation*}
		\begin{split}
			\Tr\Big[D_{x\log(x)}^{[2]}(\rho_0,\nu)\Big]=&\Tr\Big[\nu L_{1+\log(x)}(\rho_0,\nu)\Big]\\
			=&\Tr\Big[\nu L_{\log(x)}(\rho_0,\nu)\Big].
		\end{split}
	\end{equation*}
\end{proof}

This second-order perturbation theory for the von Neumann entropy that gives the same result as Eq.~\ref{eq:VonNeumannEntropy} has been developed previously \cite{Chen2010,Rodrigues2019}. However, our expression is more compact, it is straightforwardly derived from matrix calculus, and it elegantly connects the second-order behavior of the entropy to the geometry of the perturbing operator $\nu$ on the zeroth-order state $\rho_0$ through concept of the Fr\'echet derivative.

In the next three theorems we find series expansions for three different quantities that relate two quantum states $\rho_1$ and $\rho_2$ that differ by two different small perturbations $\nu_1$ and $\nu_2$. We find that, similarly to the von Neumann entropy, second-order expansions for distance measures between perturbed states can be analytically evaluated without diagonalization of the perturbations, requiring only the eigenvalues and eigenvectors of $\rho_0$ and the matrix elements $\mel{\phi_k}{\nu_1-\nu_2}{\phi_l}$ in the eigenbasis of $\rho_0$. We will need the following lemmas.

\begin{lemma}
	Let $\rho_0=\sum_i \lambda_i \outerproduct{\phi_i}{\phi_i}$ and $\nu$ be operators on $\mathcal{H}$, where the support of $\nu$ is a subspace of the support of $\rho_0$. If $f(x)$ is a primary matrix function that is $C^1$ at the eigenvalues $\vec{\lambda}$, 
	\begin{equation}
		\Tr\Big[f'(\rho_0)^{-1}L_{f(x)}(\rho_0,\nu)\Big]=\Tr[\nu].
		\label{eq:Frechet1LemmaDistances}
	\end{equation}
	\label{lem:Frechet1LemmaDistances}
\end{lemma}

\begin{proof}
	The proof of Eq.~\ref{eq:Frechet1LemmaDistances} parallels the proof given for Lemma~\ref{lem:Frechet2LemmaEntropies}:
	\begin{equation*}
		\begin{split}
			\Tr\big[f'(\rho_0)^{-1}L_{f(x)}(\rho_0,\nu)\big]=&\Tr\Big[\sum_{i}f'(\lambda_i)^{-1}\outerproduct{\phi_i}{\phi_i}\sum_{k,l}[f(x),\vec{\lambda}]^{[1]}_{k,l}\mel{\phi_k}{\nu}{\phi_l}\outerproduct{\phi_k}{\phi_l}\Big]\\
			=&\Tr\Big[\sum_{k}f'(\lambda_k)^{-1}f'(\lambda_k)\mel{\phi_k}{\nu}{\phi_k}\outerproduct{\phi_k}{\phi_k}\Big]\\
			=&\Tr[\nu].
		\end{split}
	\end{equation*}
\end{proof}

\begin{lemma}
	Let $\rho_0=\sum_i \lambda_i \outerproduct{\phi_i}{\phi_i}$ and $\nu$ be operators on $\mathcal{H}$, where the support of $\nu$ is a subspace of the support of $\rho_0$. If $f(x)$ is a primary matrix function that is $C^2$ at the eigenvalues $\vec{\lambda}$, 
	\begin{equation}
		\Tr\Big[f'(\rho_0)^{-1}D_{f(x)}^{[2]}(\rho_0,\nu)\Big]=-\Tr\Big[L_{f'(x)^{-1}}(\rho_0,\nu) L_{f(x)}(\rho_0,\nu)\Big].
		\label{eq:Frechet2LemmaDistances}
	\end{equation}
	\label{lem:Frechet2LemmaDistances}
\end{lemma}

\begin{proof}
	Applying Eqs.~\ref{eq:QuantumDaleckii-Krein2}, \ref{eq:FrechetDerivativeQuantum} and \ref{eq:2ndDerivativeQuantum}, we have
	\begin{equation*}
		\begin{split}
			\Tr\Big[f'(\rho_0)^{-1}D_{f(x)}^{[2]}(\rho_0,\nu)\Big]=&2\Tr\Big[\sum_{i}\frac{1}{f'(\lambda_i)}\outerproduct{\phi_i}{\phi_i}\sum_{k,l,m} [f(x),\vec{\lambda}]_{k,l,m}^{[2]}\mel{\phi_k}{\nu}{\phi_l}\mel{\phi_l}{\nu}{\phi_m}\outerproduct{\phi_k}{\phi_m}\Big]\\
			=&2\sum_{k,l} \frac{1}{f'(\lambda_k)}[f(x),\vec{\lambda}]_{k,l,k}^{[2]}\abs{\mel{\phi_k}{\nu}{\phi_l}}^2\\
			=&2\bigg(\sum_{k<l}\Big(\frac{1}{f'(\lambda_k)}[f(x),\vec{\lambda}]_{k,l,k}^{[2]}+\frac{1}{f'(\lambda_l)}[f(x),\vec{\lambda}]_{l,k,l}^{[2]}\Big)\abs{\mel{\phi_k}{\nu}{\phi_l}}^2+\sum_{l}\frac{1}{f'(\lambda_l)}[f(x),\vec{\lambda}]_{l,l,l}^{[2]}\abs{\mel{\phi_l}{\nu}{\phi_l}}^2\bigg).
		\end{split}
	\end{equation*}
	Here the terms containing derivatives are canceled in the first sum, which were the ones that remained in the proof of Lemma~\ref{lem:Frechet2LemmaEntropies}. We then use the identity $\partial/\partial\lambda_l f'(\lambda_l)^{-1}=f''(\lambda_l)/f'(\lambda_l)^2$ to find
	\begin{equation*}
		\begin{split}
			\Tr\Big[f'(\rho_0)^{-1}D_{f(x)}^{[2]}(\rho_0,\nu)\Big]=&2\bigg(\sum_{k<l}\Big(-\frac{f'(\lambda_k)^{-1}}{\lambda_k-\lambda_l}-\frac{f'(\lambda_l)^{-1}}{\lambda_l-\lambda_k}\Big)[f(x),\vec{\lambda}]_{k,l}^{[1]}\abs{\mel{\phi_k}{\nu}{\phi_l}}^2+\frac{1}{2}\sum_{l}\frac{f''(\lambda_l)}{f'(\lambda_l)}\abs{\mel{\phi_l}{\nu}{\phi_l}}^2\bigg)\\
			=&-\sum_{k,l}[f'(x)^{-1},\vec{\lambda}]_{k,l}^{[1]}[f(x),\vec{\lambda}]_{k,l}^{[1]}\abs{\mel{\phi_k}{\nu}{\phi_l}}^2\\
			=&\Tr\Big[\sum_{i,j}[f'(x)^{-1},\vec{\lambda}]_{k,l}^{[1]}\mel{\phi_i}{\nu}{\phi_j}\outerproduct{\phi_i}{\phi_j}\sum_{k,l} [f(x),\vec{\lambda}]_{k,l}^{[1]}\mel{\phi_k}{\nu}{\phi_l}\outerproduct{\phi_k}{\phi_l}\Big]\\
			=&\Tr\Big[L_{f'(x)^{-1}}(\rho_0,\nu) L_{f(x)}(\rho_0,\nu)\Big].
		\end{split}
	\end{equation*}
\end{proof}

\begin{lemma}
	Let $\rho_0=\sum_i \lambda_i \outerproduct{\phi_i}{\phi_i}$, and let $\nu_1$ and $\nu_2$ be operators on $\mathcal{H}$, where the support of $\nu_1$ and $\nu_2$ are each subspaces of the support of $\rho_0$. If $f(x)$ and $g(x)$ are primary matrix functions that are $C^1$ at the eigenvalues $\vec{\lambda}$, 
	\begin{equation}
		\Tr\big[L_{f(x)}(\rho_0,\nu_1) L_{g(x)}(\rho_0,\nu_1)-2L_{f(x)}(\rho_0,\nu_1) L_{g(x)}(\rho_0,\nu_2)+L_{f(x)}(\rho_0,\nu_2) L_{g(x)}(\rho_0,\nu_2)\big]=\Tr\big[L_{f(x)}(\rho_0,\nu_1-\nu_2) L_{g(x)}(\rho_0,\nu_1-\nu_2)\big].
		\label{eq:Frechet1Lemmafg}
	\end{equation}
	\label{lem:Frechet1Lemmafg}
\end{lemma}

\begin{proof}
	We start by explicitly writing out the matrix elements of the middle term of Eq.~\ref{eq:Frechet1Lemmafg} using Eq.~\ref{eq:FrechetDerivativeQuantum}:
	\begin{equation*}
		\begin{split}
			\Tr\big[L_{f(x)}(\rho_0,\nu_1) L_{g(x)}(\rho_0,\nu_2)\big]=&\Tr\Big[\sum_{i,j}[f(x),\vec{\lambda}]_{i,j}^{[1]}\mel{\phi_i}{\nu_1}{\phi_j}\outerproduct{\phi_i}{\phi_j}\sum_{k,l}[g(x),\vec{\lambda}]_{k,l}^{[1]}\mel{\phi_k}{\nu_2}{\phi_l}\outerproduct{\phi_k}{\phi_l}\Big]\\
			=&\sum_{k,l}[f(x),\vec{\lambda}]_{l,k}^{[1]}[g(x),\vec{\lambda}]_{k,l}^{[1]}\mel{\phi_l}{\nu_1}{\phi_k}\mel{\phi_k}{\nu_2}{\phi_l}\\
			=&\sum_{k,l}[f(x),\vec{\lambda}]_{l,k}^{[1]}[g(x),\vec{\lambda}]_{k,l}^{[1]}\mel{\phi_l}{\nu_2}{\phi_k}\mel{\phi_k}{\nu_1}{\phi_l}\\
			=&\Tr\big[L_{f(x)}(\rho_0,\nu_2)L_{g(x)}(\rho_0,\nu_1)\big],
		\end{split}
	\end{equation*}
	where in the third line we swap the indices $k$ and $l$ and use the symmetry of the divided difference matrix (Eq.~\ref{eq:DividedDifference1}). We can then prove Eq.~\ref{eq:Frechet1Lemmafg}:
	\begin{equation*}
		\begin{split}
			\Tr\big[&L_{f(x)}(\rho_0,\nu_1) L_{g(x)}(\rho_0,\nu_1)-2L_{f(x)}(\rho_0,\nu_1) L_{g(x)}(\rho_0,\nu_2)+L_{f(x)}(\rho_0,\nu_2) L_{g(x)}(\rho_0,\nu_2)\big]\\
			=&\Tr\big[L_{f(x)}(\rho_0,\nu_1) L_{g(x)}(\rho_0,\nu_1)-L_{f(x)}(\rho_0,\nu_1) L_{g(x)}(\rho_0,\nu_2)-L_{f(x)}(\rho_0,\nu_2) L_{g(x)}(\rho_0,\nu_1)+L_{f(x)}(\rho_0,\nu_2) L_{g(x)}(\rho_0,\nu_2)\big]\\
			=&\sum_{k,l}[f(x),\vec{\lambda}]_{l,k}^{[1]}[g(x),\vec{\lambda}]_{k,l}^{[1]}\Big(\abs{\mel{\phi_k}{\nu_1}{\phi_l}}^2-\mel{\phi_l}{\nu_1}{\phi_k}\mel{\phi_k}{\nu_2}{\phi_l}-\mel{\phi_l}{\nu_2}{\phi_k}\mel{\phi_k}{\nu_1}{\phi_l}+\abs{\mel{\phi_k}{\nu_2}{\phi_l}}^2\Big)\\
			=&\sum_{k,l}[f(x),\vec{\lambda}]_{l,k}^{[1]}[g(x),\vec{\lambda}]_{k,l}^{[1]}\abs{\mel{\phi_k}{\nu_1-\nu_2}{\phi_l}}^2\\
			=&\Tr\big[L_{f(x)}(\rho_0,\nu_1-\nu_2) L_{g(x)}(\rho_0,\nu_1-\nu_2)\big].
		\end{split}
	\end{equation*}
\end{proof}

We begin with the quantum relative entropy of $\rho_1$ with respect to $\rho_2$. It is important to note that the QRE is only well defined when $\textrm{supp}(\rho_1)\subseteq\textrm{supp}(\rho_2)$ \cite{Wilde2013}, and we assume this condition is satisfied throughout the remainder of the paper.
\begin{theorem}
	Let $\rho_1=\rho_0+\nu_1$ and $\rho_2=\rho_0+\nu_2$ be two unit-trace quantum states on $\mathcal{H}$, where $\nu_1$ and $\nu_2$ are each an arbitrary zero-trace state perturbation on $\mathcal{H}$ with support that is a subspace of the support of $\rho_0$. The quantum relative entropy of $\rho_1$ with respect to $\rho_2$ is
	\begin{equation}
		D(\rho_1||\rho_2)= \frac{1}{2}\Tr[(\nu_1-\nu_2)L_{\log(x)}(\rho_0,\nu_1-\nu_2)]+O\big(\max(\norm{\nu_1},\norm{\nu_2})^3\big).
		\label{eq:RelativeEntropySupportPreserving}
	\end{equation}
	\label{thm:RelativeEntropySupportPreserving}
\end{theorem}

\begin{proof}
	We first apply Theorem~\ref{thm:QuantumDaleckii-Krein2} twice to the definition of the QRE with $f(x)=\log(x)$ in Eq.~\ref{eq:RelativeEntropy}:
	\begin{equation*}
		\begin{split}
			D(\rho_1||\rho_2)=&\Tr\bigg[\rho_1\Big(\log(\rho_0)+L_{\log(x)}(\rho_0,\nu_1)+\frac{1}{2}D_{\log(x)}^{[2]}(\rho_0,\nu_1)+O(\norm{\nu_1}^3)\\
			&-\log(\rho_0)-L_{\log(x)}(\rho_0,\nu_2)-\frac{1}{2}D_{\log(x)}^{[2]}(\rho_0,\nu_2)-O(\norm{\nu_2}^3)\Big)\bigg]\\
			=&\Tr\bigg[\rho_1\bigg(L_{\log(x)}(\rho_0,\nu_1-\nu_2)+\frac{1}{2}\Big(D_{\log(x)}^{[2]}(\rho_0,\nu_1)-D_{\log(x)}^{[2]}(\rho_0,\nu_2)\Big)\bigg)\bigg]+O\big(\max(\norm{\nu_1},\norm{\nu_2})^3\big)\\
			=&\Tr\bigg[\nu_1L_{\log(x)}(\rho_0,\nu_1-\nu_2)+\frac{1}{2}\rho_0\Big(D^{[2]}_{\log(x)}(\rho_0,\nu_1)-D^{[2]}_{\log(x)}(\rho_0,\nu_2)\Big)\bigg]+O\big(\max(\norm{\nu_1},\norm{\nu_2})^3\big),
		\end{split}
	\end{equation*}
	where in the last line we use Lemma~\ref{lem:Frechet1LemmaDistances} and recall that $\Tr[\nu_1]=\Tr[\nu_2]=0$ to cancel the term $\Tr[\rho_0L_{\log(x)}(\rho_0,\nu_1-\nu_2)]=\Tr[\nu_1-\nu_2]$, and where  $\Tr\big[\nu_1\big(D^{[2]}_{\log(x)}(\rho_0,\nu_1)-D^{[2]}_{\log(x)}(\rho_0,\nu_2)\big)\big]$ is $O\big(\max(\norm{\nu_1},\norm{\nu_2})^3\big)$. After two applications of Lemma~\ref{lem:Frechet2LemmaDistances}, we find
	\begin{equation*}
		\begin{split}
			D(\rho_1||\rho_2)=&\Tr\bigg[\nu_1L_{\log(x)}(\rho_0,\nu_1-\nu_2)-\frac{1}{2}\Big(\nu_1L_{\log(x)}(\rho_0,\nu_1)-\nu_2L_{\log(x)}(\rho_0,\nu_2)\Big)\bigg]+O\big(\max(\norm{\nu_1},\norm{\nu_2})^3\big)\\
			=&\frac{1}{2}\Tr\bigg[\nu_1L_{\log(x)}(\rho_0,\nu_1)-2\nu_1L_{\log(x)}(\rho_0,\nu_2)+\nu_2L_{\log(x)}(\rho_0,\nu_2)\bigg]+O\big(\max(\norm{\nu_1},\norm{\nu_2})^3\big)\\
			=&\frac{1}{2}\Tr[(\nu_1-\nu_2)L_{\log(x)}(\rho_0,\nu_1-\nu_2)]+O\big(\max(\norm{\nu_1},\norm{\nu_2})^3\big),
		\end{split}
	\end{equation*}
	where in the last line we use Lemma~\ref{lem:Frechet1Lemmafg} with $f(x)=x$ and $g(x)=\log(x)$.
\end{proof}

We confirm that the first non-zero term in the expansion for the quantum relative entropy, unlike the von Neumann entropy, is second-order \cite{Rodrigues2019}. The perturbations appear in the second-order expansion only through the difference $\nu_1-\nu_2$, reflecting a sort of relativity for the QRE; adding a small, constant operator to both of the two perturbations will not change the QRE. Furthermore, it is easy to see that this second-order term is symmetric between $\rho_1$ and $\rho_2$, which is not true of the QRE for two general quantum states. 

\begin{theorem}
	Let $\rho_1=\rho_0+\nu_1$ and $\rho_2=\rho_0+\nu_2$ be two unit-trace quantum states on $\mathcal{H}$, where $\nu_1$ and $\nu_2$ are each an arbitrary zero-trace state perturbation on $\mathcal{H}$ with support that is a subspace of the support of $\rho_0$. The quantum Chernoff bound for a binary hypothesis test between $\rho_1$ and $\rho_2$ is found by maximizing
	\begin{equation}
		\xi_s(\rho_1,\rho_2)= \frac{1}{2}\Tr\big[L_{x^s}(\rho_0,\nu_1-\nu_2)L_{x^{1-s}}(\rho_0,\nu_1-\nu_2)\big]+O\big(\max(\norm{\nu_1},\norm{\nu_2})^3\big),
		\label{eq:ChernoffInfoSupportPreserving}
	\end{equation}
	over $s\in[0,1]$, where the maximum is found at $s=1/2$ and is given by
	\begin{equation}
		\xi(\rho_1,\rho_2)=\max_{s\in[0,1]}\xi_s(\rho_1,\rho_2)=\xi_{1/2}(\rho_1,\rho_2)= \frac{1}{2}\Tr\big[L_{\sqrt{x}}(\rho_0,\nu_1-\nu_2)^2\big]+O\big(\max(\norm{\nu_1},\norm{\nu_2})^3\big).
		\label{eq:BhattacharyyaInfoSupportPreserving}
	\end{equation}
	\label{thm:ChernoffInfoSupportPreserving}
\end{theorem}

\begin{proof}
	Beginning with Eq.~\ref{eq:ChernoffInformation}, we apply Theorem~\ref{thm:QuantumDaleckii-Krein2} twice to both $\rho_1^s$ and $\rho_2^{1-s}$:
	\begin{equation*}
		\begin{split}
			\xi_s(\rho_1,\rho_2)=&-\log\bigg(\Tr\Big[\big(\rho_0^s+L_{x^s}(\rho_0,\nu_1)+\frac{1}{2}D_{x^s}^{[2]}(\rho_0,\nu_1)+O\big(\norm{\nu_1}^3\big)\big)\\
			&\times\big(\rho_0^{1-s}+L_{x^{1-s}}(\rho_0,\nu_2)+\frac{1}{2}D_{x^{1-s}}^{[2]}(\rho_0,\nu_2)+O\big(\norm{\nu_2}^3\big)\big)\Big]\bigg)\\
			=&-\log\bigg(\Tr\Big[\rho_0+\rho_0^sL_{x^{1-s}}(\rho_0,\nu_2)+L_{x^s}(\rho_0,\nu_1)\rho_0^{1-s}+L_{x^s}(\rho_0,\nu_1)L_{x^{1-s}}(\rho_0,\nu_2)\\
			&+\frac{1}{2}\rho_0^sD_{x^{1-s}}^{[2]}(\rho_0,\nu_2)+\frac{1}{2}D_{x^s}^{[2]}(\rho_0,\nu_1)\rho_0^{1-s}+O\big(\max(\norm{\nu_1},\norm{\nu_2})^3\big)\Big]\bigg)\\
			=&-\log\bigg(\Tr\Big[\rho_0+s\nu_1+(1-s)\nu_2+L_{x^s}(\rho_0,\nu_1)L_{x^{1-s}}(\rho_0,\nu_2)\\
			&-\frac{1}{2}(1-s)L_{x^s/(1-s)}(\rho_0,\nu_2)L_{x^{1-s}}(\rho_0,\nu_2)-\frac{1}{2}sL_{x^{1-s}/s}(\rho_0,\nu_1)L_{x^s}(\rho_0,\nu_1)+O\big(\max(\norm{\nu_1},\norm{\nu_2})^3\big)\Big]\bigg),
		\end{split}
	\end{equation*}
	where the last equality uses two applications each of Lemma~\ref{lem:Frechet1LemmaDistances} and Lemma~\ref{lem:Frechet2LemmaDistances}. Using $\Tr[\rho_0]=1$, $\Tr[\nu_1]=\Tr[\nu_2]=0$, $L_{af(x)}(A,E)=aL_{f(x)}(A,E)$, and Lemma~\ref{lem:Frechet1Lemmafg}, we find
	\begin{equation*}
		\xi_s(\rho_1,\rho_2)=-\log\Big(1-\frac{1}{2}\Tr\big[L_{x^s}(\rho_0,\nu_1-\nu_2)L_{x^{1-s}}(\rho_0,\nu_1-\nu_2)\big]+O\big(\max(\norm{\nu_1},\norm{\nu_2})^3\big)\Big),
	\end{equation*}
	and by $\log(1-x)=x+O(x^2)$ we arrive at Eq.~\ref{eq:ChernoffInfoSupportPreserving}.
	
	To verify Eq.~\ref{eq:BhattacharyyaInfoSupportPreserving}, we consider the trace in Eq.~\ref{eq:ChernoffInfoSupportPreserving}:
	\begin{equation*}
		\Tr\Big[L_{x^s}(\rho_0,\nu_1-\nu_2)L_{x^{1-s}}(\rho_0,\nu_1-\nu_2)\Big]=\sum_{k,l}[x^s,\vec{\lambda}]_{k,l}^{[1]}[x^{1-s},\vec{\lambda}]_{l,k}^{[1]}\abs{\mel{\phi_k}{\nu_1-\nu_2}{\phi_l}}^2\outerproduct{\phi_k}{\phi_k}.
	\end{equation*}
	It will be sufficient to prove the proposition $[x^s,\vec{\lambda}]_{k,l}^{[1]}[x^{1-s},\vec{\lambda}]_{l,k}^{[1]}\leq\big([\sqrt{x},\vec{\lambda}]_{k,l}^{[1]}\big)^2$ for all $s\in[0,1]$ and all $k$ and $l$. When $\lambda_k\neq\lambda_l$,
	\begin{equation*}
		[x^s,\vec{\lambda}]_{k,l}^{[1]}[x^{1-s},\vec{\lambda}]_{l,k}^{[1]}=\frac{\lambda_k^s-\lambda_l^s}{\lambda_k-\lambda_l}\frac{\lambda_l^{1-s}-\lambda_k^{1-s}}{\lambda_l-\lambda_k}=\frac{\lambda_k+\lambda_l-\lambda_k^s\lambda_l^{1-s}-\lambda_l^s\lambda_k^{1-s}}{(\lambda_k-\lambda_l)^2},
	\end{equation*}
	whereas when $\lambda_k=\lambda_l$,
	\begin{equation*}
		[x^s,\vec{\lambda}]_{k,l}^{[1]}[x^{1-s},\vec{\lambda}]_{l,k}^{[1]}=s\lambda_k^{s-1}(1-s)\lambda_k^{-s}=\frac{s(1-s)}{\lambda_k}.
	\end{equation*}
	The proposition is then proven because $a^sb^{1-s}+b^sa^{1-s}\geq2\sqrt{ab}$ and  $s(1-s)\leq1/4$ for $a\in\mathbb{R}$, $b\in\mathbb{R}$, $s\in[0,1]$.
\end{proof}

From Theorem~\ref{thm:ChernoffInfoSupportPreserving}, it follows that the QCB for binary discrimination between states separated by support-preserving perturbations is always saturated by the generally looser quantum Bhattacharyya bound \cite{Pirandola2008}, which removes the need for a minimization over the parameter $s$ and simplifies the computation of the QCB. The second-order expression for the QCB also exhibits a relativity between $\nu_1$ and $\nu_2$, as the only dependence on the perturbations appears in the difference $\nu_1-\nu_2$.

\begin{theorem}
	Let $\rho_1=\rho_0+\nu_1$ and $\rho_2=\rho_0+\nu_2$ be two unit-trace quantum states on $\mathcal{H}$, where $\nu_1$ and $\nu_2$ are each an arbitrary zero-trace state perturbation on $\mathcal{H}$ with support that is a subspace of the support of $\rho_0$. The quantum fidelity between $\rho_1$ and $\rho_2$ is
	\begin{equation}
		F(\rho_1,\rho_2)= 1-\frac{1}{2}\Tr\big[(\nu_1-\nu_2)L_{\sqrt{x}}(\rho_0^2,\nu_1-\nu_2)\big]+O\big(\max(\norm{\nu_1},\norm{\nu_2})^3\big).
		\label{eq:FidelitySupportPreserving}
	\end{equation}
	\label{thm:FidelitySupportPreserving}
\end{theorem}

\begin{proof}
	It will be useful to find simplified forms for the first divided differences of $f(x)=\sqrt{x}$ and $g(x)=x^{-1/2}$. From Eq.~\ref{eq:DividedDifference1}, the terms of the first divided difference matrices are
	\begin{eqnarray}
		\label{eq:sqrtxDividedDifference1kl} [\sqrt{x},\vec{\alpha}]_{k,l}^{[1]}&=&\frac{\sqrt{\alpha_k}-\sqrt{\alpha_l}}{\alpha_k-\alpha_l}=\frac{1}{\sqrt{\alpha_k}+\sqrt{\alpha_l}}\\
		\label{eq:invsqrtxDividedDifference1kl} [x^{-1/2},\vec{\alpha}]_{k,l}^{[1]}&=&\frac{\frac{1}{\sqrt{\alpha_k}}-\frac{1}{\sqrt{\alpha_l}}}{\alpha_k-\alpha_l}=-\frac{1}{\sqrt{\alpha_k\alpha_l}(\sqrt{\alpha_k}+\sqrt{\alpha_l})}
	\end{eqnarray}
	when $\alpha_k\neq\alpha_l$ and
	\begin{eqnarray}
		\label{eq:sqrtxDividedDifference1kk} [\sqrt{x},\vec{\alpha}]_{k,l}^{[1]}=\frac{d}{d \alpha_k} \sqrt{\alpha_k}&=&\frac{1}{2\sqrt{\alpha_k}}\\
		\label{eq:invsqrtxDividedDifference1kk} [x^{-1/2},\vec{\alpha}]_{k,l}^{[1]}=\frac{d}{d \alpha_k} \frac{1}{\sqrt{\alpha_k}}&=&-\frac{1}{2\alpha_k^{3/2}}
	\end{eqnarray}
	when $\alpha_k=\alpha_l$. From these observations we can rewrite the first divided difference matrices for all $k$ and $l$ as
	\begin{eqnarray}
		\label{eq:sqrtxDividedDifference1klUnified} [\sqrt{x},\vec{\alpha}]_{k,l}^{[1]}&=&\frac{1}{\sqrt{\alpha_k}+\sqrt{\alpha_l}}\\
		\label{eq:invsqrtxDividedDifference1klUnified} [x^{-1/2},\vec{\alpha}]_{k,l}^{[1]}&=&-\frac{1}{\sqrt{\alpha_k\alpha_l}(\sqrt{\alpha_k}+\sqrt{\alpha_l})}.
	\end{eqnarray}

	Using Theorem~\ref{thm:QuantumDaleckii-Krein2} we can write
	\begin{equation*}
		\begin{split}
			\sqrt{\rho_1}\rho_2\sqrt{\rho_1}=&\bigg(\sqrt{\rho_0}+L_{\sqrt{x}}(\rho_0,\nu_1)+\frac{1}{2}D_{\sqrt{x}}^{[2]}(\rho_0,\nu_2)+O(\norm{\nu_1}^3)\bigg)\big(\rho_0+\nu_2\big)\bigg(\sqrt{\rho_0}+L_{\sqrt{x}}(\rho_0,\nu_1)+\frac{1}{2}D_{\sqrt{x}}^{[2]}(\rho_0,\nu_2)+O(\norm{\nu_1}^3)\bigg)\\
			=&\rho_0^2+\bar{\nu}+\bar{\bar{\nu}}+O\big(\max(\norm{\nu_1},\norm{\nu_2})^3\big),
		\end{split}
	\end{equation*}
	where
	\begin{equation*}
		\begin{split}
			\bar{\nu}=&\rho_0^{3/2}L_{\sqrt{x}}(\rho_0,\nu_1)+L_{\sqrt{x}}(\rho_0,\nu_1)\rho_0^{3/2}+\sqrt{\rho_0}\nu_2\sqrt{\rho_0}\\
			=&\sum_{k,l}\bigg(\frac{\lambda_k^{3/2}+\lambda_l^{3/2}}{\sqrt{\lambda_k}+\sqrt{\lambda_l}}\mel{\phi_k}{\nu_1}{\phi_l}+\sqrt{\lambda_k\lambda_l}\mel{\phi_k}{\nu_2}{\phi_l}\bigg)\outerproduct{\phi_k}{\phi_l}\\
			=&\sum_{k,l}\big((\lambda_k+\lambda_l)\mel{\phi_k}{\nu_1}{\phi_l}+\sqrt{\lambda_k\lambda_l}\mel{\phi_k}{\nu_2-\nu_1}{\phi_l}\big)\outerproduct{\phi_k}{\phi_l}\\
			=&\rho_0\nu_1+\nu_1\rho_0+\sqrt{\rho_0}(\nu_2-\nu_1)\sqrt{\rho_0}
		\end{split}
	\end{equation*}
	is $O\big(\max(\norm{\nu_1},\norm{\nu_2})\big)$ and where
	\begin{equation*}
		\begin{split}
			\bar{\bar{\nu}}=&\rho_0^{3/2}D_{\sqrt{x}}^{[2]}(\rho_0,\nu_1)+D_{\sqrt{x}}^{[2]}(\rho_0,\nu_1)\rho_0^{3/2}+L_{\sqrt{x}}(\rho_0,\nu_1)\rho_0L_{\sqrt{x}}(\rho_0,\nu_1)+\sqrt{\rho_0}\nu_2L_{\sqrt{x}}(\rho_0,\nu_1)+L_{\sqrt{x}}(\rho_0,\nu_1)\nu_2\sqrt{\rho_0}
		\end{split}
	\end{equation*}
	is $O\big(\max(\norm{\nu_1},\norm{\nu_2})^2\big)$. We then use Theorem~\ref{thm:QuantumDaleckii-Krein2} again to expand $\Tr\big[\sqrt{\sqrt{\rho_1}\rho_2\sqrt{\rho_1}}\big]$ as
	\begin{equation}
		\begin{split}
			\Tr\bigg[\sqrt{\sqrt{\rho_1}\rho_2\sqrt{\rho_1}}\bigg]=&\Tr\bigg[\rho_0+L_{\sqrt{x}}\Big(\rho_0^2,\bar{\nu}+\bar{\bar{\nu}}\Big)+\frac{1}{2}D_{\sqrt{x}}^{[2]}\Big(\rho_0^2,\bar{\nu}+\bar{\bar{\nu}}\Big)+O\big(\max(\norm{\nu_1},\norm{\nu_2})^3\big)\bigg]\\
			=&1+\Tr\bigg[L_{\sqrt{x}}\Big(\rho_0^2,\bar{\nu}+\bar{\bar{\nu}}\Big) + \frac{1}{4}\bar{\nu}L_{1/\sqrt{x}}\Big(\rho_0^2,\bar{\nu}\Big)\bigg]+O\big(\max(\norm{\nu_1},\norm{\nu_2})^3\big),
		\end{split}
		\label{eq:FidelitySupportPreservingIntermediate}
	\end{equation}
	where Lemma~\ref{lem:Frechet2LemmaEntropies} and some additional grouping of $O\big(\max(\norm{\nu_1},\norm{\nu_2})^3\big)$ terms were used to reach the expression in the last line. 
	
	We consider the two remaining terms within the trace in Eq.~\ref{eq:FidelitySupportPreservingIntermediate} individually. Applying Lemma~\ref{lem:Frechet1LemmaEntropies} and then Lemma~\ref{lem:Frechet2LemmaDistances} to the first term, we have
	\begin{equation*}
		\begin{split}
			\Tr\Big[L_{\sqrt{x}}\Big(\rho_0^2,\bar{\nu}+\bar{\bar{\nu}}\Big)\Big]=&\frac{1}{2}\Tr\Big[(\bar{\nu}+\bar{\bar{\nu}})\rho_0^{-1}\Big]\\
			=&\frac{1}{2}\Tr\Big[\nu_1+\nu_2+2\sqrt{\rho_0}D_{\sqrt{x}}^{[2]}(\rho_0,\nu_1)+L_{\sqrt{x}}(\rho_0,\nu_1)\rho_0L_{\sqrt{x}}(\rho_0,\nu_1)\rho_0^{-1}\\
			&+\rho_0^{-1/2}\big(\nu_2L_{\sqrt{x}}(\rho_0,\nu_1)+L_{\sqrt{x}}(\rho_0,\nu_1)\nu_2\big)\Big]\\
			=&\frac{1}{2}\Tr\Big[-L_{\sqrt{x}}(\rho_0,\nu_1)L_{\sqrt{x}}(\rho_0,\nu_1)+L_{\sqrt{x}}(\rho_0,\nu_1)\rho_0L_{\sqrt{x}}(\rho_0,\nu_1)\rho_0^{-1}\\
			&+\rho_0^{-1/2}\big(\nu_2L_{\sqrt{x}}(\rho_0,\nu_1)+L_{\sqrt{x}}(\rho_0,\nu_1)\nu_2\big)\Big].
		\end{split}
	\end{equation*}
	We use Lemma~\ref{lem:Frechet1Lemmafg} to evaluate the second term to
	\begin{equation*}
		\begin{split}
			\Tr\bigg[\frac{1}{4}\bar{\nu}L_{1/\sqrt{x}}\Big(\rho_0^2,\bar{\nu}\Big)\bigg]=&\frac{1}{4}\Tr\Big[(\rho_0\nu_1+\nu_1\rho_0)L_{1/\sqrt{x}}\Big(\rho_0^2,\rho_0\nu_1+\nu_1\rho_0\Big)+2\sqrt{\rho_0}(\nu_2-\nu_1)\sqrt{\rho_0}L_{1/\sqrt{x}}\Big(\rho_0^2,\rho_0\nu_1+\nu_1\rho_0\Big)\\
			&+\sqrt{\rho_0}(\nu_2-\nu_1)\sqrt{\rho_0}L_{1/\sqrt{x}}\Big(\rho_0^2,\sqrt{\rho_0}(\nu_2-\nu_1)\sqrt{\rho_0}\Big)\bigg]
		\end{split}
	\end{equation*}
	
	After using Eq.~\ref{eq:FrechetDerivativeQuantum} to rewrite the Fr\'echet derivatives, the two terms become
	\begin{equation*}
		\begin{split}
			\Tr\Big[L_{\sqrt{x}}\Big(\rho_0^2,\bar{\nu}+\bar{\bar{\nu}}\Big)\Big]=&\frac{1}{2}\Tr\bigg[\sum_{k,l,m}\frac{1}{(\sqrt{\lambda_k}+\sqrt{\lambda_l})(\sqrt{\lambda_l}+\sqrt{\lambda_m})}\bigg(-1+\frac{\lambda_l}{\lambda_m}\bigg)\mel{\phi_k}{\nu_1}{\phi_l}\mel{\phi_l}{\nu_1}{\phi_m}\outerproduct{\phi_k}{\phi_m}\\
			&+\frac{1}{\sqrt{\lambda_k}}\bigg(\frac{\mel{\phi_k}{\nu_2}{\phi_l}\mel{\phi_l}{\nu_1}{\phi_m}}{\sqrt{\lambda_l}+\sqrt{\lambda_m}}+\frac{\mel{\phi_k}{\nu_1}{\phi_l}\mel{\phi_l}{\nu_2}{\phi_m}}{\sqrt{\lambda_k}+\sqrt{\lambda_l}}\bigg)\outerproduct{\phi_k}{\phi_m}\bigg]\\
			=&\frac{1}{2}\Tr\bigg[\sum_{k,l}\frac{1}{\lambda_k}\bigg(\frac{\sqrt{\lambda_l}-\sqrt{\lambda_k}}{\sqrt{\lambda_k}+\sqrt{\lambda_l}}\abs{\mel{\phi_k}{\nu_1}{\phi_l}}^2+\frac{\sqrt{\lambda_k}}{\sqrt{\lambda_k}+\sqrt{\lambda_l}}\big(\mel{\phi_k}{\nu_2}{\phi_l}\mel{\phi_l}{\nu_1}{\phi_k} + \textrm{c.c.}\big)\bigg)\outerproduct{\phi_k}{\phi_k}\bigg]\\
			=&\frac{1}{2}\Tr\bigg[\bigg(\sum_{k<l}\frac{\lambda_l-\lambda_k}{\lambda_l\lambda_k}\frac{\sqrt{\lambda_l}-\sqrt{\lambda_k}}{\sqrt{\lambda_k}+\sqrt{\lambda_l}}\abs{\mel{\phi_k}{\nu_1}{\phi_l}}^2+2\frac{\sqrt{\lambda_k}+\sqrt{\lambda_l}}{\sqrt{\lambda_k\lambda_l}}\frac{1}{\sqrt{\lambda_k}+\sqrt{\lambda_l}}\mel{\phi_k}{\nu_2}{\phi_l}\mel{\phi_l}{\nu_1}{\phi_k}\\
			&+\sum_{k}\frac{1}{\lambda_k}\mel{\phi_k}{\nu_2}{\phi_k}\mel{\phi_k}{\nu_1}{\phi_k}\bigg)\outerproduct{\phi_k}{\phi_k}\bigg]\\
			=&\frac{1}{4}\Tr\bigg[\sum_{k,l}\bigg(\frac{(\sqrt{\lambda_l}-\sqrt{\lambda_k})^2}{\lambda_l\lambda_k}\abs{\mel{\phi_k}{\nu_1}{\phi_l}}^2+\frac{2}{\sqrt{\lambda_k\lambda_l}}\mel{\phi_k}{\nu_2}{\phi_l}\mel{\phi_l}{\nu_1}{\phi_k}\bigg)\outerproduct{\phi_k}{\phi_k}\bigg]
		\end{split}
	\end{equation*}
	and
	\begin{equation*}
		\begin{split}
			\Tr\bigg[\frac{1}{4}\bar{\nu}L_{1/\sqrt{x}}\Big(\rho_0^2,\bar{\nu}\Big)\bigg]=&\frac{1}{4}\Tr\bigg[\sum_{k,l,m}\frac{(\lambda_k+\lambda_l)(\lambda_l+\lambda_m)}{-\lambda_l\lambda_m(\lambda_l+\lambda_m)}\mel{\phi_k}{\nu_1}{\phi_l}\mel{\phi_l}{\nu_1}{\phi_m}\outerproduct{\phi_k}{\phi_m}\\
			&+2\frac{\sqrt{\lambda_k\lambda_l}(\lambda_l+\lambda_m)}{-\lambda_l\lambda_m(\lambda_l+\lambda_m)}\mel{\phi_k}{\nu_2-\nu_1}{\phi_l}\mel{\phi_l}{\nu_1}{\phi_m}\outerproduct{\phi_k}{\phi_m}\\
			&+\frac{\sqrt{\lambda_k\lambda_l}\sqrt{\lambda_l\lambda_m}}{-\lambda_l\lambda_m(\lambda_l+\lambda_m)}\mel{\phi_k}{\nu_2-\nu_1}{\phi_l}\mel{\phi_l}{\nu_2-\nu_1}{\phi_m}\outerproduct{\phi_k}{\phi_m}\bigg]\\
			=&-\frac{1}{4}\Tr\Bigg[\sum_{k,l}\Bigg(\bigg(\frac{\lambda_k+\lambda_l}{\lambda_k\lambda_l}-2\frac{1}{\sqrt{\lambda_k\lambda_l}}\bigg)\abs{\mel{\phi_k}{\nu_1}{\phi_l}}^2\\
			&+\frac{2}{\sqrt{\lambda_k\lambda_l}}\mel{\phi_k}{\nu_2}{\phi_l}\mel{\phi_l}{\nu_1}{\phi_k}+\frac{1}{\lambda_k+\lambda_l}\abs{\mel{\phi_k}{\nu_2-\nu_1}{\phi_l}}^2\Bigg)\outerproduct{\phi_k}{\phi_k}\Bigg]\\
			=&-\frac{1}{4}\Tr\bigg[\sum_{k,l}\bigg(\frac{(\sqrt{\lambda_l}-\sqrt{\lambda_k})^2}{\lambda_l\lambda_k}\abs{\mel{\phi_k}{\nu_1}{\phi_l}}^2+\frac{2}{\sqrt{\lambda_k\lambda_l}}\mel{\phi_k}{\nu_2}{\phi_l}\mel{\phi_l}{\nu_1}{\phi_k}\\
			&+\frac{1}{\lambda_k+\lambda_l}\abs{\mel{\phi_k}{\nu_2-\nu_1}{\phi_l}}^2\Bigg)\outerproduct{\phi_k}{\phi_k}\Bigg].
		\end{split}
	\end{equation*}
	Inserting the two terms into Eq.~\ref{eq:FidelitySupportPreservingIntermediate}, we find 
	\begin{equation*}
		\begin{split}
			\Tr\bigg[\sqrt{\sqrt{\rho_1}\rho_2\sqrt{\rho_1}}\bigg]=&1-\frac{1}{4}\Tr\bigg[\sum_{k,l}\frac{1}{\lambda_k+\lambda_l}\abs{\mel{\phi_k}{\nu_2-\nu_1}{\phi_l}}^2\outerproduct{\phi_k}{\phi_k}\bigg]+O\big(\max(\norm{\nu_1},\norm{\nu_2})^3\big)\\
			=&1-\frac{1}{4}\Tr\Big[(\nu_2-\nu_1)L_{\sqrt{x}}(\rho_0^2,\nu_2-\nu_1)\Big]+O\big(\max(\norm{\nu_1},\norm{\nu_2})^3\big),
		\end{split}
	\end{equation*}
	and, after using the Taylor expansion $(1+x)^2 = 1+2x+O(x^2)$ to evaluate $\Tr\big[\sqrt{\sqrt{\rho_1}\rho_2\sqrt{\rho_1}}\big]^2$, we arrive at Eq.~\ref{eq:FidelitySupportPreserving}.
\end{proof}

As expected, the fidelity between two states separated by small support-preserving perturbations evaluates to unity to zeroth-order in the Hilbert-Schmidt norms of the perturbing operators, and the next lowest-order term is second order. This second-order term only depends on the perturbations through the difference $\nu_1-\nu_2$. Furthermore, if we explicitly write out the Fr\'echet derivative and evaluate the Bures distance $d_{\rm B}^2(\rho_1,\rho_2)=2(1-\sqrt{F(\rho_1,\rho_2)})$, we find
\begin{equation}
	\begin{split}
		d_{\rm B}^2(\rho_1,\rho_2)=&\frac{1}{2}\sum_{k,l}\frac{\abs{\mel{\phi_k}{\nu_1-\nu_2}{\phi_l}}^2}{\lambda_k+\lambda_l}+O\big(\max(\norm{\nu_1},\norm{\nu_2})^3\big)\\
		=&\frac{1}{2}\sum_{k,l}\frac{\abs{\mel{\phi_k}{\rho_1-\rho_2}{\phi_l}}^2}{\lambda_k+\lambda_l}+O\big(\max(\norm{\nu_1},\norm{\nu_2})^3\big).
	\end{split}
	\label{eq:BuresMetricSupportPreserving}
\end{equation}
Setting $\rho_1=\rho$ and $\rho_2=\rho+d\rho$, our perturbation theory for the quantum fidelity exactly recovers the expression for the Bures distance of a density operator with respect to a support-preserving perturbation $d\rho$ \cite{Hubner1992}. In addition, if $\rho_1=\rho_{\vec{X}}$ and $\rho_2=\rho_{\vec{X}+d\vec{X}}$, the Bures distance can be used to derive the quantum Fisher information matrix (QFIM) for estimation of parameter(s) $\vec{X}$ \cite{Braunstein1994}.

\section{Support-Extending Perturbation Theory}
\label{sec:Results_Support-Extending}

When a matrix $A$ has a kernel $\mathcal{H}_0$ that is a subspace of $\mathcal{H}$, it is possible that its primary matrix functions will not be Fr\'echet differentiable in the direction of matrices whose support has an intersection with $\mathcal{H}_0$. This scenario is relevant in the case of a matrix root $f(x)=x^{1/p}$, $p>0$, when the support of a state perturbation $\nu$ extends beyond that of the zeroth-order state $\rho_0$; in this case, the Dalecki\u{\i}-Kre\u{\i}n expansion of Theorem~\ref{thm:QuantumDaleckii-Krein2} does not apply directly. However, we now show that for all of the quantum information theoretic quantities we considered in the case of support-preserving perturbations, lowest-order series expansions can be obtained using a first-order Dalecki\u{\i}-Kre\u{\i}n-like expansion for the roots of perturbations of singular matrices, given in Eq.~\ref{eq:Daleckii-Krein_Singular} \cite{Carlsson2018}. This expansion can be readily applied for a quantum state $\rho_0$ on a Hilbert space $\mathcal{H}=\mathcal{H}_+\oplus\mathcal{H}_0$, with support on only $\mathcal{H}_+$, and a state perturbation $\nu$ on $\mathcal{H}$. Let the full vector of eigenvalues of $\rho_0$ be given by $\vec{\lambda}$ and denote the nonzero eigenvalues by $\vec{\lambda}_+$. Following Eq.~\ref{eq:BlockDecompositionE}, we decompose the representations of $\rho_0$ and $\nu$ in an eigenbasis of $\rho_0$ as
\begin{eqnarray}
	\label{eq:BlockDecomposition_rho_0} \rho_0&=&
	\begin{pmatrix}
		\Lambda_{\vec{\lambda}_+} & 0\\
		0 & 0
	\end{pmatrix}\\
	\label{eq:BlockDecomposition_nu} \nu&=&
	\begin{pmatrix}
		\nu_{\rm B} & \nu_{\rm C}\\
		\nu_{\rm C}^{\dagger} & \nu_{\rm D}
	\end{pmatrix}.
\end{eqnarray}
Since $\rho=\rho_0+\nu$ must be Hermitian, positive semi-definite and unit-trace, $\nu_{\rm B}$ and $\nu_{\rm D}$ must be Hermitian matrices, $\Tr[\nu_{\rm B}]=-\Tr[\nu_{\rm D}]$, and $\nu_{\rm D}$ must be positive semi-definite. The following theorem provides a first-order expansion for roots of perturbed quantum states in the case of support-extending perturbations.

\begin{theorem}
	Let $1/3<s<1$, where $r=\min(1+s,3s)$. Consider a unit-trace density operator $\rho_0=\sum_i \lambda_i \outerproduct{\phi_i}{\phi_i}$ on $\mathcal{H}=\mathcal{H}_+\oplus\mathcal{H}_0$ with  $\ker(\rho_0)=\mathcal{H}_0$ and a zero-trace perturbation operator $\nu$ on $\mathcal{H}$ with $\supp(\nu)\supseteq~\mathcal{H}_0$, which are decomposed according to Eqs.~\ref{eq:BlockDecomposition_rho_0} and~\ref{eq:BlockDecomposition_nu}, respectively. If $\norm{\nu}^2/\norm{\nu_{\rm D}}\in O(\norm{\nu})$, then a first-order expansion in $\norm{\nu}$ for matrix roots of the state $\rho=\rho_0+\nu$ represented in an eigenbasis of $\rho_0$ is given by
	\begin{equation}
		\rho^s = \rho_0^s + \big[x^s,\vec{\lambda}\big]^{[1,0]}\circ
		\begin{pmatrix}
			\nu_{\rm B} & \nu_{\rm C}\\
			\nu_{\rm C}^{\dagger} & \nu_{\rm D}^s
		\end{pmatrix}
		+ O(\norm{\nu}^r).
		\label{eq:QuantumDaleckii-Krein1SupportExtending}
	\end{equation}
	\label{thm:QuantumDaleckii-Krein1SupportExtending}
\end{theorem}

\begin{proof}
	Comparing Eqs.~\ref{eq:Daleckii-Krein_Singular} and \ref{eq:QuantumDaleckii-Krein1SupportExtending}, the only thing to be proven is the replacement of $\bar{\nu}_{\rm D}=\nu_{\rm D}-\nu_{\rm C}^{\dagger}(\Lambda_{\vec{\lambda}_+} +\nu_{\rm B})^{-1}\nu_{\rm C}$ with $\nu_{\rm D}$ in the lower right block of the second term. It is sufficient to prove that $\bar{\nu}_{\rm D}^s=\nu_{\rm D}^s+O(\norm{\nu}^{1+s})$, since $1+s\geq r$. Temporarily reverting to the notation of Eqs.~\ref{eq:BlockDecompositionA} and~\ref{eq:BlockDecompositionE}, we examine more closely the Schur complement $\bar{D}=D-C^{\dagger}(\Lambda_{\vec{\alpha}_+}+B)^{-1}C=\norm{D}\big(D/\norm{D}-C^{\dagger}(\Lambda_{\vec{\alpha}_+}+B)^{-1}C/\norm{D}\big)$. Since $\nu_{\rm D}$ is full rank on $\mathcal{H}_0$, the zeroth-order Dalecki\u{\i}-Kre\u{\i}n expansion (Eq.~\ref{eq:Daleckii-Krein0}) gives $\bar{D}^s=D^s+\norm{D}^sO\big(\norm{C^{\dagger}(\Lambda_{\vec{\alpha}_+}+B)^{-1}C}/\norm{D}\big)$. Furthermore,
	\begin{equation}
		\begin{split}
			\norm{C^{\dagger}(\Lambda_{\vec{\alpha}_+}+B)^{-1}C}\leq&\norm{C}^2\norm{(\Lambda_{\vec{\alpha}_+}+B)^{-1}}\\
			=&\norm{C}^2\big(\norm{\Lambda_{\vec{\alpha}_+}^{-1}+O(\norm{B})}\big)\\
			=&\norm{C}^2\big(\norm{\Lambda_{\vec{\alpha}_+}^{-1}+O(\norm{E})}\big)\\
			\leq&\norm{E}^2\big(\norm{\Lambda_{\vec{\alpha}_+}^{-1}+O(\norm{E})}\big)\\
			\leq&\norm{E}^2\big(\norm{\Lambda_{\vec{\alpha}_+}^{-1}}+\norm{O(\norm{E})}\big)\\
			\leq&\norm{E}^2\big(\norm{\Lambda_{\vec{\alpha}_+}^{-1}}+1\big).
		\end{split}
		\label{eq:Schur_Complement_Inequalities}
	\end{equation}
	where the first inequality follows from the fact that the Hilbert-Schmidt norm is submultiplicative, the first equality is another use of the zeroth-order Daleckii-Kre\u{\i}n theorem, the second equality and second inequality make use of a compression inequality on Schatten norms for block partitioned positive semidefinite matrices \cite{Audenaert2005}, which states that  $\norm{\hat{E}}^2=2\norm{C}^2+\norm{B}^2+\norm{D}^2$ and therefore $\norm{B}\leq\norm{\hat{E}}$ and $\norm{C}\leq\norm{\hat{E}}$, while $\norm{\hat{E}}=\norm{E}$ because the Hilbert-Schmidt norm is conserved under unitary rotation. The third inequality is the triangle inequality. Returning to the notation of Eqs.~\ref{eq:BlockDecomposition_rho_0} and~\ref{eq:BlockDecomposition_nu}, we thus have
	\begin{equation*}
		\begin{split}
			\bar{\nu}_{\rm D}^s=&\nu_{\rm D}^s+O\bigg(\frac{\norm{\nu}^2}{\norm{\nu_{\rm D}}^{1-s}}\bigg)\\
			=&\nu_{\rm D}^s+O\Bigg(\norm{\nu}^{2s}\bigg(\frac{\norm{\nu}^2}{\norm{\nu_{\rm D}}}\bigg)^{1-s}\Bigg)\\
			=&\nu_{\rm D}^s+O\big(\norm{\nu}^{1+s}\big)
		\end{split}
	\end{equation*} 
	where we used the condition $\norm{\nu}^2/\norm{\nu_{\rm D}}\in O(\norm{\nu})$ for the final equality.
\end{proof}

Next we prove that requiring the perturbing matrix to span $\mathcal{H}_0$ while not canceling out any of the eigenvalues of $A$ on $\mathcal{H}_+$ allows for a first-order series expansion over all real-valued matrix roots instead of the subset of matrix roots allowed by the existing result (Eq.~\ref{eq:Daleckii-Krein_Singular}). The following theorem relies on identifying the remainder terms present at each block of the matrix decomposition on $\mathcal{H}=\mathcal{H}_+\oplus\mathcal{H}_0$. We return to the notation of a zeroth-order matrix $A$ and a perturbation matrix $E$ for notational convenience; the theorem can be straightforwardly applied to density matrices using the notation of Eqs.~\ref{eq:BlockDecomposition_rho_0} and~\ref{eq:BlockDecomposition_nu}.

\begin{theorem}
	Let $A$ be a Hermitian matrix on $\mathcal{H}=\mathcal{H}_+\oplus \mathcal{H}_0$ with $\ker(A)=\mathcal{H}_0$ and let $E$ be a second Hermitian matrix on $\mathcal{H}$ that are decomposed according to Eqs.~\ref{eq:BlockDecompositionA} and~\ref{eq:BlockDecompositionE}, respectively. If $\supp(E)\supseteq~\mathcal{H}_0$, $\supp(A+E)\supseteq~\mathcal{H}_+$, and $\norm{E}^2/\norm{D}\in O(\norm{E})$, then for all $0\leq s\leq1$,
	\begin{equation}
		(A+E)^s = A^s + U\bigg(\big[x^s,\vec{\alpha}\big]^{[1,0]}\circ
		\begin{pmatrix}
			B & C\\
			C^{\dagger} & D^s
		\end{pmatrix}
		+
		\begin{pmatrix}
			O\big(\norm{E}^{2}\big) & O\big(\norm{E}^{1+s}\big)\\
			O\big(\norm{E}^{1+s}\big) & O\big(\norm{E}^{1+s}\big)
		\end{pmatrix}
		\bigg)U^{\dagger}.
		\label{eq:Carlsson_extension}
	\end{equation}
	\label{thm:Carlsson_extension}
\end{theorem}

\begin{proof}
	We derive the remainder term of Eq.~\ref{eq:Carlsson_extension} using similar methods to those used to prove Theorem 3.1 in Ref.~\cite{Carlsson2018}. 
	Working in the eigenbasis of $A$, it was proved in Ref.~\cite{Carlsson2018} using a first-order Dalecki\u{\i}-Kre\u{\i}n expansion (Eq.~\ref{eq:Daleckii-Krein1}) that the three quadrants apart from the lower right block on the right hand side of Eq.~\ref{eq:Carlsson_extension} are entirely determined (apart from the remainder terms) by the Fr\'{e}chet differentiable quantity
	\begin{equation}
		\begin{pmatrix}
			\Lambda_{\vec{\alpha}_+}+B & C\\
			C^{\dagger} & C^{\dagger}\big(\Lambda_{\vec{\alpha}_+}+B\big)^{-1}C
		\end{pmatrix}
		^s
		=\Lambda_{\vec{\alpha}}^s+\big[x^s,\vec{\alpha}\big]^{[1,0]}\circ
		\begin{pmatrix}
			B & C\\
			C^{\dagger} & 0
		\end{pmatrix}
		+O\big(\norm{E}^2\big).
		\label{eq:Carlsson_Extension_Intermediate1}
	\end{equation}
	Therefore, verifying the statement
	\begin{equation}
		\begin{pmatrix}
			\Lambda_{\vec{\alpha}_+}+B & C\\
			C^{\dagger} & C^{\dagger}\big(\Lambda_{\vec{\alpha}_+}+B\big)^{-1}C+\bar{D}
		\end{pmatrix}
		^s
		=
		\begin{pmatrix}
			\Lambda_{\vec{\alpha}_+}+B & C\\
			C^{\dagger} & C^{\dagger}\big(\Lambda_{\vec{\alpha}_+}+B\big)^{-1}C
		\end{pmatrix}
		^s
		+
		\begin{pmatrix}
			0 & 0 \\
			0 & \bar{D}^s
		\end{pmatrix}
		+
		\begin{pmatrix}
			O\big(\norm{E}^{2+s}\big) & O\big(\norm{E}^{1+s}\big)\\
			O\big(\norm{E}^{1+s}\big) &
			O\big(\norm{E}^{1+s}\big)
		\end{pmatrix}
		\label{eq:Carlsson_Extension_Intermediate2}
	\end{equation}
	will confirm the remainder terms since the left hand side is equal to $\big(\Lambda_{\vec{\alpha}}+\hat{E}\big)^s$ and since $2+s>2$. Applying a spectral decomposition to the first term on the right hand side of Eq.~\ref{eq:Carlsson_Extension_Intermediate2} sans the exponent gives
	\begin{equation*}
		\begin{pmatrix}
			\Lambda_{\vec{\alpha}_+}+B & C\\
			C^{\dagger} & C^{\dagger}\big(\Lambda_{\vec{\alpha}_+}+B\big)^{-1}C
		\end{pmatrix}
		=
		V
		\begin{pmatrix}
			\Lambda_{\vec{\alpha}_1^+} & 0\\
			0 & 0
		\end{pmatrix}
		V^{\dagger},
		\label{eq:Carlsson_Extension_Intermediate3}
	\end{equation*}
	where $\vec{\alpha}_1^+$ are the nonzero eigenvalues of $A_1=\Lambda_{\vec{\alpha}}+\tilde{E}$ and $\tilde{E}=
	\begin{pmatrix} 
		B & C\\
		C^{\dagger} & C^{\dagger}\big(\Lambda_{\vec{\alpha}_+}+B\big)^{-1}C
	\end{pmatrix}
	$. It is easy to work out that one choice for the diagonalization of $A_1$ is $W=
	\begin{pmatrix}
		\mathcal{I} & -\big(\Lambda_{\vec{\alpha}_+}+B\big)^{-1}C\\
		C^{\dagger}\big(\Lambda_{\vec{\alpha}_+}+B\big)^{-1} & \mathcal{I}
	\end{pmatrix}
	=
	\begin{pmatrix}
		\mathcal{I} & O(\norm{E})\\
		O(\norm{E}) & \mathcal{I}
	\end{pmatrix}$, and a Gram-Schmidt orthogonalization procedure (see the proof of Lemma 3.3 in Ref.~\cite{Carlsson2018}) yields the unitary matrix of orthogonal eigenvectors $V=W+O(\norm{E}^2)$. With this choice, we define $B_1$ and $C_1$ such that 
	\begin{equation*}
		V^{\dagger}
		\begin{pmatrix}
			0 & 0\\
			0 & \bar{D}
		\end{pmatrix}
		V
		=\begin{pmatrix}
			B_1 & C_1\\
			C_1^{\dagger} & D_1
		\end{pmatrix}
	\end{equation*}
	where $B_1=O\big(\norm{E}^3\big)$, $C_1=O\big(\norm{E}^2\big)$, and $D_1=\bar{D}+O\big(\norm{E}^2\big)$ by inspection, and we consider the matrix
	\begin{equation*}
		\begin{split}
			\frac{1}{\norm{E}}\begin{pmatrix} 
				\Lambda_{\vec{\alpha}_1^+}+B_1 & C_1\\ C_1^{\dagger} & D_1
			\end{pmatrix}
			=&\tilde{A}_1/\norm{E}
			+\tilde{E}_1/\norm{E}.
		\end{split}
	\end{equation*}
	where $\tilde{A}_1=\begin{pmatrix} \Lambda_{\vec{\alpha}_1^+} & 0 \\ 0 & \bar{D} \end{pmatrix}$ and $\tilde{E}_1=
	\begin{pmatrix} O(\norm{E}^3) & O(\norm{E}^2) \\ O(\norm{E}^2) & O(\norm{E}^2) \end{pmatrix}$. Since $\tilde{A}_1$ is full rank on $\mathcal{H}$, we can use a first-order Dalecki\u{\i}-Kre\u{\i}n expansion (Eq.~\ref{eq:Daleckii-Krein1}) to write
	\begin{equation*}
		\begin{split}
			\big(\tilde{A}_1/{\norm{E}}+\tilde{E}_1/{\norm{E}}\big)^s=&\big(\tilde{A}_1/\norm{E}\big)^s+\big[x^s,\vec{\alpha}_1\big]^{[1]}\circ \big(\tilde{E}_1/\norm{E}\big)+O(\norm{\tilde{E}/\norm{E}}^2)\\
			=&
			\begin{pmatrix}
				\Lambda_{\vec{\alpha}_1^+}^s/\norm{E}^s+O\big(\norm{E}^2\big) & O(\norm{E})\\
				O(\norm{E}) & \bar{D}^s/\norm{E}^s + O\big(\norm{E}\big)
			\end{pmatrix}
		\end{split}
	\end{equation*}
	
	Returning to the expression on the left hand side of Eq.~\ref{eq:Carlsson_Extension_Intermediate2}, we have
	\begin{equation*}
		\begin{split}
			\begin{pmatrix}
				\Lambda_{\vec{\alpha}_+}+B & C\\
				C^{\dagger} & C^{\dagger}\big(\Lambda_{\vec{\alpha}_+}+B\big)^{-1}C+\bar{D}
			\end{pmatrix}
			^s
			=&\Bigg(V
			\begin{pmatrix} 
				\Lambda_{\vec{\alpha}_1^+}+B_1 & C_1\\ 
				C_1^{\dagger} & D_1
			\end{pmatrix}
			V^{\dagger}\bigg)^s\\
			=&V\norm{E}^s
			\begin{pmatrix}
				\Lambda_{\vec{\alpha}_1^+}^s/\norm{E}^s+O\big(\norm{E}^{2}\big) & O(\norm{E})\\
				O(\norm{E}) & \bar{D}^s/\norm{E}^s+O\big(\norm{E}\big)
			\end{pmatrix}
			V^{\dagger}\\
			=&
			\bigg(V
			\begin{pmatrix}
				\Lambda_{\vec{\alpha}_1^+} & 0\\
				0 & 0
			\end{pmatrix}
			V^{\dagger}\bigg)^s+V
			\begin{pmatrix}
				0 & 0 \\
				0 & \bar{D}^s
			\end{pmatrix}
			V^{\dagger}+
			V
			\begin{pmatrix}
				O\big(\norm{E}^{2+s}\big) & O(\norm{E}^{1+s})\\
				O(\norm{E}^{1+s}) & O\big(\norm{E}^{1+s}\big)
			\end{pmatrix}
			V^{\dagger}.
		\end{split}
	\end{equation*}
	Using Eq.~\ref{eq:Carlsson_Extension_Intermediate3} for the first term on the right hand side of the final expression and using $V=\begin{pmatrix} \mathcal{I}+O(\norm{E}^2) & O(\norm{E})\\ O(\norm{E}) & \mathcal{I}+O(\norm{E}^2) \end{pmatrix}$ for the second and third terms results in the equality of Eq.~\ref{eq:Carlsson_Extension_Intermediate2}. Finally, a similar analysis to that from the proof of Theorem~\ref{thm:QuantumDaleckii-Krein1SupportExtending} can be performed to replace $\bar{D}^s$ in the lower right matrix block of Eq.~\ref{eq:Carlsson_Extension_Intermediate2} with $D^s$, finishing the proof.
\end{proof} 

We now find first-order expansions for our list of quantum information theoretic quantities (Table~\ref{tab:MainResults}) about small support-extending perturbations.

\begin{theorem}
	Let $\rho_0=\sum_i \lambda_i \outerproduct{\phi_i}{\phi_i}$ be a unit-trace quantum state on $\mathcal{H}=\mathcal{H}_+\oplus\mathcal{H}_0$ with support on $\mathcal{H}_+$, and let $\nu$ be a zero-trace state perturbation on $\mathcal{H}$ with $\norm{\nu}\ll1$ and a decomposition given in Eq.~\ref{eq:BlockDecomposition_nu}. If $\norm{\nu}^2/\norm{\nu_{\rm D}}\in O(\norm{\nu})$, the von Neumann entropy of the state $\rho=\rho_0+\nu$ is given by
	\begin{equation}
		S(\rho)=S(\rho_0)-\Tr\big[L_{x\log(x)}(\rho_{0_+},\nu_{\rm B})\big]-\Tr[\nu_{\rm D}\log(\nu_{\rm D})] +O(\norm{\nu}^2),
		\label{eq:VonNeumannEntropySupportExtending}
	\end{equation}
	where the Fr\'echet derivative $L_{x\log(x)}(\rho_{0_+},\nu_{\rm B})=\sum_{k,l} [x\log(x),\vec{\lambda}_+]^{[1]}_{k,l} \mel{\phi_k}{\nu_{\rm B}}{\phi_l}\outerproduct{\phi_k}{\phi_l}$ is evaluated only on the subspace $\mathcal{H}_+$.
	\label{thm:VonNeumannEntropySupportExtending}
\end{theorem}

\begin{proof}
	We prove the theorem by finding a first-order expansion in $\norm{\nu}$ for the quantum Tsallis entropy, defined as \cite{Abe2001,Hu2006}
	\begin{equation}
		S_q(\rho) = \frac{1}{1-q}\big(\Tr[\rho^q]-1\big),
		\label{eq:TsallisEntropy}
	\end{equation}
	where $q\in(0,1)\cup(1,\infty)$. We can expand the quantum Tsallis entropy using Theorem~\ref{thm:QuantumDaleckii-Krein1SupportExtending} when $1/3<q<1$ and $r=\min(1+q,3q)$:
	\begin{equation*}
		\begin{split}
			S_q(\rho)=&\frac{1}{1-q}\Bigg(\Tr\bigg[\rho_0^q + \big[x^q,\vec{\lambda}\big]^{[1,0]}\circ
			\begin{pmatrix}
				\nu_{\rm B} & \nu_{\rm C}\\
				\nu_{\rm C}^{\dagger} & \nu_{\rm D}^q
			\end{pmatrix}
			+ O(\norm{\nu}^r)\bigg]-1\Bigg)\\
			=&\frac{1}{q}\big(\Tr[\rho_0^q]-1\big)+\frac{1}{1-q}\Bigg(\Tr\bigg[\big[x^q,\vec{\lambda}\big]^{[1,0]}\circ
			\begin{pmatrix}
				\nu_{\rm B} & \nu_{\rm C}\\
				\nu_{\rm C}^{\dagger} & \nu_{\rm D}^q
			\end{pmatrix}
			\bigg]\Bigg)+ O(\norm{\nu}^r)\\
			=&S_q(\rho_0)+\frac{1}{1-q}\bigg(\Tr\Big[\big[x^q,\vec{\lambda}_+\big]^{[1]}\circ \nu_{\rm B}\Big]+\Tr\big[\nu_{\rm D}^q\big]\bigg)+ O(\norm{\nu}^r)\\
			=&S_q(\rho_0)+\frac{1}{1-q}\bigg(\Tr\Big[q\Lambda_{\vec{\alpha}_+}^{q-1}\circ \nu_{\rm B}\Big]+\Tr\big[\nu_{\rm D}^q\big]\bigg)+ O(\norm{\nu}^r)\\
			=&S_q(\rho_0)+\frac{1}{1-q}\bigg(\Tr\Big[(q-1)\Lambda_{\vec{\alpha}_+}^{q-1}\circ \nu_{\rm B}\Big]+\Tr\Big[(\Lambda_{\vec{\alpha}_+}^{q-1}-\mathcal{I})\circ \nu_{\rm B}\Big]+\Tr\Big[\nu_{\rm B}\Big]+\Tr\big[\nu_{\rm D}^q\big]\bigg)+ O(\norm{\nu}^r)\\
			=&S_q(\rho_0)-\Tr\Big[\Lambda_{\vec{\alpha}_+}^{q-1}\circ \nu_{\rm B}\Big]+\Tr\bigg[\frac{1}{1-q}(\Lambda_{\vec{\alpha}_+}^{q-1}-\mathcal{I})\circ \nu_{\rm B}\bigg]+\Tr\bigg[\frac{1}{1-q}(\nu_{\rm D}^q-\nu_{\rm D})\bigg]+ O(\norm{\nu}^r),
		\end{split}
	\end{equation*}
	where in the last line we use the fact that $\Tr[\nu]=0$ implies $\Tr[\nu_{\rm B}]=-\Tr[\nu_{\rm D}]$. The von Neumann entropy can be related to the quantum Tsallis entropy by $S(\rho)=\lim_{q\to1}S_q(\rho)$ \cite{Abe2001,Hu2006}. Noting that $\lim_{q\to1}r=2$, the von Neumann entropy can be written as
	\begin{equation*}
		\begin{split}
			S(\rho)=&\lim_{q\to1}S_q(\rho)\\
			&=S(\rho) - \Tr[\nu_{\rm B}] + \Tr\bigg[\lim_{q\to1}\frac{1}{1-q}(\Lambda_{\vec{\alpha}_+}^{q-1}-\mathcal{I})\circ \nu_{\rm B}\bigg]+\Tr\bigg[\lim_{q\to1}\frac{1}{1-q}\nu_{\rm D}(\nu_{\rm D}^{q-1}-\mathcal{I})\bigg]+ O(\norm{\nu}^2)\\
			=&S(\rho) - \Tr[\nu_{\rm B}] + \Tr\bigg[\sum_{k=1}^{\dim(\mathcal{H}_+)}\mel{\phi_k}{\nu_{\rm B}}{\phi_k}\lim_{q\to1}\frac{\lambda_k^{q-1}-1}{1-q}\bigg]
			+\Tr\bigg[\sum_{i=1}^{\dim(\mathcal{H}_0)}\mel{\phi_i}{\nu_{\rm D}}{\phi_i}\lim_{q\to1}\frac{\mel{\phi_i}{\nu_{\rm D}}{\phi_i}^{q-1}-1}{1-q}\bigg]+ O(\norm{\nu}^2)\\
			=&S(\rho) - \Tr[\nu_{\rm B}] - \Tr\bigg[\sum_{k=1}^{\dim(\mathcal{H}_+)}\mel{\phi_k}{\nu_{\rm B}}{\phi_k}\log(\lambda_k)\bigg]
			-\Tr\bigg[\sum_{i=1}^{\dim(\mathcal{H}_0)}\mel{\phi_i}{\nu_{\rm D}}{\phi_i}\log(\mel{\phi_i}{\nu_{\rm D}}{\phi_i})\bigg]+ O(\norm{\nu}^2)\\
			=&S(\rho) - \Tr\Big[\Big(\mathcal{I}+\log\big(\Lambda_{\vec{\lambda}_+}\big)\Big)\circ\nu_{\rm B}\Big] - \Tr\big[\nu_{\rm D}\log(\nu_{\rm D})\big]+ O(\norm{\nu}^2)\\
			=&S(\rho) - \Tr\big[L_{x\log(x)}(\rho_0,\nu_{\rm B})\big] - \Tr\big[\nu_{\rm D}\log(\nu_{\rm D})\big]+ O(\norm{\nu}^2),
		\end{split}
	\end{equation*}
	where we twice used the identity $\lim_{q\to1} (a^{q-1}-1)/(q-1)=\log(a)$.
\end{proof}

The first-order correction to the von Neumann entropy is a sum of a Fr\'echet derivative of $x\log(x)$ on $\mathcal{H}_+$ and a term that takes a form similar to a von Neumann entropy of the contribution of the perturbation on $\mathcal{H}_0$; however, while $\nu_{\rm D}$ is a positive semi-definite and Hermitian, it is not normalized to unit trace and therefore $S(\nu_{\rm D})\neq-\Tr[\nu_{\rm D}\log(\nu_{\rm D})]$.

\begin{theorem}
	Let $\rho_0=\sum_i \lambda_i \outerproduct{\phi_i}{\phi_i}$ be a unit-trace quantum state on $\mathcal{H}=\mathcal{H}_+\oplus\mathcal{H}_0$ with support on $\mathcal{H}_+$, and let $\nu_1$ and $\nu_2$ be zero-trace state perturbations on $\mathcal{H}$ with $\nu_1\supseteq\mathcal{H}_0$, $\nu_2\supseteq\mathcal{H}_0$, $\nu_1\subseteq\nu_2$, and given suitable decompositions for $\nu_1$ and $\nu_2$ according to Eq.~\ref{eq:BlockDecomposition_nu}. If $\norm{\nu}^2/\norm{\nu_{\rm D}}\in O(\norm{\nu})$, the quantum relative entropy of the state $\rho_1=\rho_0+\nu_1$ with respect to the state $\rho_2=\rho_0+\nu_2$ is given by
	\begin{equation}
		D(\rho_1||\rho_2)=\Tr[\nu_{1,\rm B}-\nu_{2,\rm B}]+\Tr\big[\nu_{1,\rm D}\big(\log(\nu_{1,\rm D})-\log(\nu_{2,\rm D})\big)\big] + O\big(\max(\norm{\nu_1},\norm{\nu_2})^2\big).
	\end{equation}
	\label{thm:RelativeEntropySupportExtending}
\end{theorem}

\begin{proof}
	In analogy to the previous proof, we find a first-order expansion in $\max(\norm{\nu_1},\norm{\nu_2})$ for the quantum Tsallis relative entropy, defined for $q\in(0,1)\cup(1,\infty)$ by \cite{Abe2003,Furuichi2004}
	\begin{equation}
		D_q(\rho_1||\rho_2)=\frac{1}{1-q}\Big(1-\Tr\big[\rho_1^q\rho_2^{1-q}\big]\Big).
		\label{eq:quantumTsallisRelativeEntropy}
	\end{equation}
	Using two applications of Theorem~\ref{thm:Carlsson_extension}, for $0<q<1$ we have
	\begin{equation*}
		\begin{split}
			\Tr\big[\rho_1^q\rho_2^{1-q}\big]=&\Tr\Bigg[\Bigg(\rho_0^q + \big[x^q,\vec{\lambda}\big]^{[1,0]}\circ
			\begin{pmatrix}
				\nu_{1,\rm B} & \nu_{1,\rm C}\\
				\nu_{1,\rm C}^{\dagger} & \nu_{1,\rm D}^q
			\end{pmatrix}
			+ \begin{pmatrix}
				O\big(\norm{\nu_1}^{2}\big) & O\big(\norm{\nu_1}^{1+q}\big)\\
				O\big(\norm{\nu_1}^{1+q}\big) & O\big(\norm{\nu_1}^{1+q}\big)
			\end{pmatrix}\Bigg)\\
			&\times\Bigg(\rho_0^{1-q} + \big[x^{1-q},\vec{\lambda}\big]^{[1,0]}\circ
			\begin{pmatrix}
				\nu_{2, \rm B} & \nu_{2, \rm C}\\
				\nu_{2, \rm C}^{\dagger} & \nu_{2, \rm D}^{1-q}
			\end{pmatrix}
			+ \begin{pmatrix}
				O\big(\norm{\nu_2}^{2}\big) & O\big(\norm{\nu_2}^{2-q}\big)\\
				O\big(\norm{\nu_2}^{2-q}\big) & O\big(\norm{\nu_2}^{2-q}\big)
			\end{pmatrix}
			\Bigg)\Bigg]\\
			=&\Tr\Bigg[\rho_0+
			\Bigg(\big[x^q,\vec{\lambda}\big]^{[1,0]}\circ
			\begin{pmatrix}
				\nu_{1,\rm B} & \nu_{1,\rm C}\\
				\nu_{1,\rm C}^{\dagger} & \nu_{1,\rm D}^q
			\end{pmatrix}\Bigg)
			\begin{pmatrix}
				\Lambda^{1-q}_{\vec{\lambda}_+}& 0\\
				0 & 0
			\end{pmatrix}
			+
			\begin{pmatrix}
				\Lambda^q_{\vec{\lambda}_+} & 0\\
				0 & 0
			\end{pmatrix}
			\Bigg(\big[x^{1-q},\vec{\lambda}\big]^{[1,0]}\circ
			\begin{pmatrix}
				\nu_{2, \rm B} & \nu_{2, \rm C}\\
				\nu_{2, \rm C}^{\dagger} & \nu_{2, \rm D}^{1-q}
			\end{pmatrix}\Bigg)
			\\
			&+\Bigg(\big[x^q,\vec{\lambda}\big]^{[1,0]}\circ
			\begin{pmatrix}
				\nu_{1,\rm B} & \nu_{1,\rm C}\\
				\nu_{1,\rm C}^{\dagger} & \nu_{1,\rm D}^q
			\end{pmatrix}
			\Bigg)
			\Bigg(\big[x^{1-q},\vec{\lambda}\big]^{[1,0]}\circ
			\begin{pmatrix}
				\nu_{2, \rm B} & \nu_{2, \rm C}\\
				\nu_{2, \rm C}^{\dagger} & \nu_{2, \rm D}^{1-q}
			\end{pmatrix}
			\Bigg)\\
			&+
			\begin{pmatrix}
				O\big(\max(\norm{\nu_1},\norm{\nu_2})^2\big) & O\big(\max(\norm{\nu_1},\norm{\nu_2})^{2-q}\big)\\
				O\big(\max(\norm{\nu_1},\norm{\nu_2})^{1+q}\big)  & O\big(\max(\norm{\nu_1},\norm{\nu_2})^2\big)
			\end{pmatrix} \Bigg]\\
			=&1+q\Tr[\nu_{1,\rm B}]+(1-q)\Tr[\nu_{2,\rm B}]\\
			&+\Tr\Bigg[\Bigg(\big[x^q,\vec{\lambda}\big]^{[1,0]}\circ
			\begin{pmatrix}
				\nu_{1,\rm B} & \nu_{1,\rm C}\\
				\nu_{1,\rm C}^{\dagger} & \nu_{1,\rm D}^q
			\end{pmatrix}
			\Bigg)
			\Bigg(\big[x^{1-q},\vec{\lambda}\big]^{[1,0]}\circ
			\begin{pmatrix}
				\nu_{2, \rm B} & \nu_{2, \rm C}\\
				\nu_{2, \rm C}^{\dagger} & \nu_{2, \rm D}^{1-q}
			\end{pmatrix}
			\Bigg)\Bigg]+ O\big(\max(\norm{\nu_1},\norm{\nu_2})^2\big).
		\end{split}
	\end{equation*}
	In the fourth term in the last expression, the submultiplicative property of the Hilbert-Schmidt norm and the aforementioned compression inequality \cite{Audenaert2005} indicate that all pairs of matrix blocks in the multiplication that contribute to the diagonal blocks will yield submatrices that are strictly $O\big(\max(\norm{\nu_1},\norm{\nu_2})^2\big)$ except for the multiplication of the two lower right blocks. As a result,
	\begin{equation}
		\Tr\big[\rho_1^q\rho_2^{1-q}\big]=1+q\Tr[\nu_{1,\rm B}]+(1-q)\Tr[\nu_{2,\rm B}]+\Tr[\nu_{1,\rm D}^q\nu_{2,\rm D}^{1-q}]+O\big(\max(\norm{\nu_1},\norm{\nu_2})^2\big).
		\label{eq:QuantumRelativeEntropySupportExtendingIntermediate}
	\end{equation}
	The quantum Tsallis relative entropy is then
	\begin{equation*}
		\begin{split}
			D_q(\rho_1||\rho_2)=&\frac{1}{1-q}\Big(1-1-q\Tr[\nu_{1,\rm B}]-(1-q)\Tr[\nu_{2,\rm B}]-\Tr[\nu_{1,\rm D}^q\nu_{2,\rm D}^{1-q}]+O\big(\max(\norm{\nu_1},\norm{\nu_2}^2)\big)\\
			=&\frac{1}{1-q}\Big((1-q)\big(\Tr[\nu_{1,\rm B}]-\Tr[\nu_{2,\rm B}]\big)-\Tr[\nu_{1,\rm B}]-\Tr[\nu_{1,\rm D}^q\nu_{2,\rm D}^{1-q}]\Big)+O\big(\max(\norm{\nu_1},\norm{\nu_2})^2\big)\\
			=&\Tr[\nu_{1,\rm B}]-\Tr[\nu_{2,\rm B}]+\frac{1}{1-q}\Big(\Tr[\nu_{1,\rm D}]-\Tr[\nu_{1,\rm D}^q\nu_{2,\rm D}^{1-q}]\Big)+O\big(\max(\norm{\nu_1},\norm{\nu_2})^2\big)\\
			=&\Tr[\nu_{1,\rm B}-\nu_{2,\rm B}]-\frac{1}{1-q}\Tr\big[\nu_{1,\rm D}\big(\nu_{1,\rm D}^{q-1}\nu_{2,\rm D}^{1-q}-\mathcal{I}\big)\big]+O\big(\max(\norm{\nu_1},\norm{\nu_2})^2\big),
		\end{split}
	\end{equation*}
	where we again used the property $\Tr[\nu_{\rm B}]=-\Tr[\nu_{\rm D}]$. Since the QRE can be related to the quantum Tsallis relative entropy by $D(\rho_1||\rho_2)=\lim_{q\to1}D_q(\rho_1||\rho_2)$ \cite{Abe2003,Furuichi2004}, the QRE can be written as 
	\begin{equation*}
		\begin{split}
			D(\rho_1||\rho_2)=&\lim_{q\to1}D_q(\rho_1||\rho_2)\\
			=&\Tr[\nu_{1,\rm B}-\nu_{2,\rm B}]-\lim_{q\to1}\frac{1}{1-q}\Tr\big[\nu_{1,\rm D}\big(\nu_{1,\rm D}^{q-1}\nu_{2,\rm D}^{1-q}-\mathcal{I}\big)\big]+O\big(\max(\norm{\nu_1},\norm{\nu_2})^2\big)\\
			=&\Tr[\nu_{1,\rm B}-\nu_{2,\rm B}]-\Tr\Bigg[\sum_{i=1}^{\dim(\mathcal{H}_0)}\mel{\phi_i}{\nu_{1,\rm D}}{\phi_i}\lim_{q\to1}\frac{\Big(\frac{\mel{\phi_i}{\nu_{1,\rm D}}{\phi_i}}{\mel{\phi_i}{\nu_{2,\rm D}}{\phi_i}}\Big)^{q-1}-1}{1-q}\Bigg]+O\big(\max(\norm{\nu_1},\norm{\nu_2})^2\big)\\
			=&\Tr[\nu_{1,\rm B}-\nu_{2,\rm B}]+\Tr\Bigg[\sum_{i=1}^{\dim(\mathcal{H}_0)}\mel{\phi_i}{\nu_{1,\rm D}}{\phi_i}\log\bigg(\frac{\mel{\phi_i}{\nu_{1,\rm D}}{\phi_i}}{\mel{\phi_i}{\nu_{2,\rm D}}{\phi_i}}\bigg)\Bigg]+O\big(\max(\norm{\nu_1},\norm{\nu_2})^2\big)\\
			=&\Tr[\nu_{1,\rm B}-\nu_{2,\rm B}]+\Tr\big[\nu_{1,\rm D}\big(\log(\nu_{1,\rm D})-\log(\nu_{2,\rm D})\big)\big] + O\big(\max(\norm{\nu_1},\norm{\nu_2})^2\big).
		\end{split}
	\end{equation*}
\end{proof}
Unlike the case with support-preserving perturbations, the expansion about small support-extending perturbations yields an asymmetric lowest-order expression between $\rho_1$ and $\rho_2$ for the quantum relative entropy. The term that depends on the contribution of the perturbations on $\mathcal{H}_0$ has a similar form to a QRE between the unnormalized Hermitian operators $\nu_{1,\rm D}$ and $\nu_{2, \rm D}$.

\begin{theorem}
	Let $\rho_0=\sum_i \lambda_i \outerproduct{\phi_i}{\phi_i}$ be a unit-trace quantum state on $\mathcal{H}=\mathcal{H}_+\oplus\mathcal{H}_0$ with support on $\mathcal{H}_+$, and let $\nu_1$ and $\nu_2$ be zero-trace state perturbations on $\mathcal{H}$ such that $\nu_1\supseteq\mathcal{H}_0$, $\nu_2\supseteq\mathcal{H}_0$, and $\nu_1$ and $\nu_2$ have suitable decompositions according to Eq.~\ref{eq:BlockDecomposition_nu}. If $\norm{\nu}^2/\norm{\nu_{\rm D}}\in O(\norm{\nu})$, the quantum Chernoff bound for a binary hypothesis test between the states $\rho_1=\rho_0+\nu_1$ and $\rho_2=\rho_0+\nu_2$ is given by 
	\begin{equation}
		\xi(\rho_1,\rho_2)=\max_{s\in[0,1]}\xi_s(\rho_1,\rho_2)=-\min_{s\in[0,1]}s\Tr[\nu_{1,\rm B}]+(1-s)\Tr[\nu_{2, \rm B}]+\Tr\big[\nu_{1,\rm D}^s\nu_{2, \rm D}^{1-s}\big]+O\big(\max(\norm{\nu_1},\norm{\nu_2})^2\big).
		\label{eq:QuantumChernoffInformation_SupportExtending}
	\end{equation}
	\label{thm:QuantumChernoffInformation_SupportExtending}
\end{theorem}

\begin{proof}
	Setting $q=s$, insert the proven result of Eq.~\ref{eq:QuantumRelativeEntropySupportExtendingIntermediate} into the definition of $\xi_s(\rho_1,\rho_2)$ from Eq.~\ref{eq:ChernoffInformation}. Eq.~\ref{eq:QuantumChernoffInformation_SupportExtending} is obtained by the first-order Taylor expansion $\log(1+x)=x+O(x^2)$.
\end{proof}

The first order expansion for the QCB, like the other quantities for support-extending perturbations, involves the calculation of a term $\Tr[\nu_{1 \rm D}^s\nu_{2, \rm D}]$ that requires diagonalization of the matrix blocks $\nu_{1 \rm D}$ and $\nu_{2, \rm D}$. In addition, the optimization over $s\in[0,1]$ remains, as the quantum Chernoff bound does not in general converge to the quantum Bhattacharyya bound with support-extending perturbations.

For the quantum fidelity of two quantum states perturbed by support-extending perturbations, we utilize a first order expansion of the matrix modulus $\abs{A}=\sqrt{A^{\dagger}A}$. For a matrix $X$ on $\mathcal{H}=\mathcal{H}_+ \oplus \mathcal{H}_0$ with support on $\mathcal{H}_+$ and a second (not necessarily Hermitian) square matrix $Z$ with dimension matching that of $X$, and given the decompositions 
\begin{eqnarray}
	\label{eq:BlockDecompositionX}X&=&
	\begin{pmatrix}
		\Lambda_{\vec{\sigma}_+} & 0\\
		0&0
	\end{pmatrix}
	\\
	\label{eq:BlockDecompositionZ}Z&=&
	\begin{pmatrix}
		Z_{1,1} & Z_{1,2}\\
		Z_{2,1} & Z_{2,2}
	\end{pmatrix}
\end{eqnarray}
the modulus of their sum is given in the eigenbasis of $X$ by \cite{Carlsson2018}
\begin{equation}
	\abs{X+Z}=\abs{X}+
	\begin{pmatrix}
		\big[\sqrt{x},\vec{\sigma}^2\big]^{[1]}\circ(\Lambda_{\vec{\sigma}_+}Z_{1,1}+Z_{1,1}^{\dagger}\Lambda_{\vec{\sigma}_+}) & Z_{1,2}\\
		Z_{1,2}^{\dagger} & \abs{Z_{2,2}}
	\end{pmatrix}
	+O(\norm{Z}^{3/2}).
	\label{eq:MatrixModulus}
\end{equation}

\begin{theorem}
	Let $\rho_0=\sum_i \lambda_i \outerproduct{\phi_i}{\phi_i}$ be a unit-trace quantum state on $\mathcal{H}=\mathcal{H}_+\oplus\mathcal{H}_0$ with support on $\mathcal{H}_+$, and let $\nu_1$ and $\nu_2$ be zero-trace state perturbations on $\mathcal{H}$ with $\norm{\nu_1}\ll1$ and $\norm{\nu_2}\ll1$ and suitable decompositions according to Eq.~\ref{eq:BlockDecomposition_nu}. If $\norm{\nu}^2/\norm{\nu_{\rm D}}\in O(\norm{\nu})$, the quantum fidelity between the states $\rho_1=\rho_0+\nu_1$ and $\rho_2=\rho_0+\nu_2$ is given by
	\begin{equation}
		F(\rho_1,\rho_2)= 1+\Tr[\nu_{1,\rm B}+\nu_{2,\rm B}]+2\Tr\big[\sqrt{\sqrt{\nu_{1,\rm D}}\nu_{2,\rm D}\sqrt{\nu_{1,\rm D}}}\big]+O\big(\max(\norm{\nu_1},\norm{\nu_2})^{3/2}\big).
		\label{eq:FidelitySupportExtending}
	\end{equation}
	\label{thm:FidelitySupportExtending}
\end{theorem}

\begin{proof}
	An equivalent definition to Eq.~\ref{eq:Fidelity} is $F(\rho_1,\rho_2)=\Tr\big[\abs{\sqrt{\rho_1}\sqrt{\rho_2}}\big]^2$. We use Theorem~\ref{thm:QuantumDaleckii-Krein1SupportExtending} twice to obtain

	\begin{equation*}
		\begin{split}
			\sqrt{\rho_1}\sqrt{\rho_2}=&\bigg(\sqrt{\rho_0}+\big[\sqrt{x},\vec{\lambda}\big]^{[1,0]}\circ
			\begin{pmatrix}
				\nu_{1,\rm B} & \nu_{1,\rm C}\\
				\nu_{1,\rm C}^{\dagger} & \sqrt{\nu_{1,\rm D}}
			\end{pmatrix}
			+ O(\norm{\nu_1}^{3/2})\bigg)\bigg(\sqrt{\rho_0}+\big[\sqrt{x},\vec{\lambda}\big]^{[1,0]}\circ
			\begin{pmatrix}
				\nu_{2,\rm B} & \nu_{2,\rm C}\\
				\nu_{2,\rm C}^{\dagger} & \sqrt{\nu_{2,\rm D}}
			\end{pmatrix}
			+ O(\norm{\nu_2}^{3/2})\bigg)\\
			=&\rho_0+\sqrt{\rho_0}\bigg(\big[\sqrt{x},\vec{\lambda}\big]^{[1,0]}\circ
			\begin{pmatrix}
				\nu_{2,\rm B} & \nu_{2,\rm C}\\
				\nu_{2,\rm C}^{\dagger} & \sqrt{\nu_{2,\rm D}}
			\end{pmatrix}\bigg)+\bigg(\big[\sqrt{x},\vec{\lambda}\big]^{[1,0]}\circ
			\begin{pmatrix}
				\nu_{1,\rm B} & \nu_{1,\rm C}\\
				\nu_{1,\rm C}^{\dagger} & \sqrt{\nu_{1,\rm D}}
			\end{pmatrix}
			\bigg)\sqrt{\rho_0}\\
			&+\bigg(\big[\sqrt{x},\vec{\lambda}\big]^{[1,0]}\circ
			\begin{pmatrix}
				\nu_{1,\rm B} & \nu_{1,\rm C}\\
				\nu_{1,\rm C}^{\dagger} & \sqrt{\nu_{1,\rm D}}
			\end{pmatrix}
			\bigg)\bigg(\big[\sqrt{x},\vec{\lambda}\big]^{[1,0]}\circ
			\begin{pmatrix}
				\nu_{2,\rm B} & \nu_{2,\rm C}\\
				\nu_{2,\rm C}^{\dagger} & \sqrt{\nu_{2,\rm D}}
			\end{pmatrix}
			\bigg)+O\big(\max(\norm{\nu_1},\norm{\nu_2})^{3/2}\big)
			\\
			=&\rho_0+
			\begin{pmatrix}
				\beta& \nu_{2, \rm C}\\
				\nu_{1, \rm C}^{\dagger}& \sqrt{\nu_{1 \rm D}}\sqrt{\nu_{2, \rm D}}
			\end{pmatrix}
			+O\big(\max(\norm{\nu_1},\norm{\nu_2})^{3/2}\big),
		\end{split}
	\end{equation*}
	where $\beta=[\sqrt{x},\vec{\lambda}]^{[1]}\circ\big(\nu_{1, \rm B}\sqrt{\Lambda_{\vec{\lambda}_+}}+\sqrt{\Lambda_{\vec{\lambda}_+}}\nu_{2, \rm B}\big)$. Using Eq.~\ref{eq:MatrixModulus}, we find
	\begin{equation*}
		\begin{split}
			\Tr\big[\abs{\sqrt{\rho_1}\sqrt{\rho_2}}\big]=&\Tr\Bigg[\abs{\rho_0}+
			\begin{pmatrix}
				[\sqrt{x},\vec{\lambda}^2]^{[1]}\circ\Big(\Lambda_{\vec{\lambda}_+}\beta+\beta^{\dagger}\Lambda_{\vec{\lambda}_+}\Big) & \nu_{2, \rm C}\\
				\nu_{2, \rm C}^{\dagger} & \abs{\sqrt{\nu_{1 \rm D}}\sqrt{\nu_{2, \rm D}}}
			\end{pmatrix}
			+O\big(\max(\norm{\nu_1},\norm{\nu_2})^{3/2}\Bigg]\\
			=&1+\Tr\Bigg[\sum_{k}\frac{1}{2\lambda_k}2\lambda_k\frac{1}{2\sqrt{\lambda_k}}\sqrt{\lambda_k}(\mel{\phi_k}{\nu_{1, \rm B}}{\phi_k}+\mel{\phi_k}{\nu_{2, \rm B}}{\phi_k})\Bigg]\\
			&+\Tr[\abs{\sqrt{\nu_{1 \rm D}}\sqrt{\nu_{2, \rm D}}}]
			+O\big(\max(\norm{\nu_1},\norm{\nu_2})^{3/2}\\
			=&1+\frac{1}{2}\Tr[\nu_{1, \rm B}+\nu_{2, \rm B}]+\Tr\bigg[\sqrt{\sqrt{\nu_{1 \rm D}}\nu_{2 \rm D}\sqrt{\nu_{1 \rm D}}}\bigg]+O\big(\max(\norm{\nu_1},\norm{\nu_2})^{3/2}.
		\end{split}
	\end{equation*}
	We arrive at Eq.~\ref{eq:FidelitySupportExtending} by using the first-order Taylor expansion $(1+x)^2=1+2x+O(x^2)$.
\end{proof}

To compute the Bures distance $d_{\rm B}^2(\rho_1,\rho_2)=2(1-\sqrt{F(\rho_1,\rho_2)})$ between two states $\rho_1$ and $\rho_2$, we evaluate 
\begin{equation}
	\begin{split}
		2(1-\sqrt{F(\rho_1,\rho_2)})=&-\Tr[\nu_{1, \rm B}+\nu_{2, \rm B}]-2\Tr\bigg[\sqrt{\sqrt{\nu_{1 \rm D}}\nu_{2 \rm D}\sqrt{\nu_{1 \rm D}}}\bigg]+O\big(\max(\norm{\nu_1},\norm{\nu_2})^{3/2}\\
		=&\Tr[\nu_{1, \rm D}+\nu_{2, \rm D}]-2\Tr\bigg[\sqrt{\sqrt{\nu_{1 \rm D}}\nu_{2 \rm D}\sqrt{\nu_{1 \rm D}}}\bigg]+O\big(\max(\norm{\nu_1},\norm{\nu_2})^{3/2}.
	\end{split}
	\label{eq:BuresMetricSupportExtending}
\end{equation}
To compute a Bures metric and derive the relationship with the quantum Fisher information \cite{Braunstein1994}, one would set $\rho_1=\rho_{\vec{X}}$ and $\rho_2=\rho_{\vec{X}+d\vec{X}}$, for which $\nu_1$ will be a matrix of all zeros on $\mathcal{H}$. If the state $\rho_{0,\vec{X}}$ is well characterized, we find a simple yet not immediately intuitive result: $d_{\rm B}^2(\rho_{0,\vec{X}},\rho_{0,\vec{X}+d\vec{X}})=\Tr[\nu_{\vec{X}+d\vec{X},\rm D}]+O(\norm{\nu_{\vec{X}+d\vec{X}}}^{3/2})$. This indicates that when a parametrized differential perturbation on a Hilbert space extends the support of a quantum state (i.e., increases its rank), the ultimate limits on the precision of an estimate of the parameter(s) governing the perturbation depend only on the quadratic rate at which probability density migrates into $\mathcal{H}_0$ from $\mathcal{H}_+$ in response to the increase in operator rank. Our perturbation theory has thus rediscovered a recent result that found that the correction term that needs to be applied in order to resolve point-like discrepancies between the QFIM and the Bures metric at locations in a state space containing discontinuities in the support of a quantum state is to take derivatives of the eigenvalues in the extended subspace \cite{Safranek2017}, in agreement with our conclusion. Our perturbation theory could be used to more accessibly investigate the relationship between the Bures metric and the QFIM, including identifying new metrics with desirable properties \cite{Zhou2019c}.

\section{Discussion}
\label{sec:Discussion}
There are other information theoretic quantities that we did not consider in this paper, including entropic quantities such as the Renyi entropy, $\alpha$-Renyi entropy, and Tsallis entropy. Similar results could be obtained for these, especially in the context of support-preserving perturbations. We do not explicitly report higher-order corrections for support-preserving perturbation theory, but such calculations can be straightforwardly inferred from Eq.~\ref{eq:Daleckii-Krein-Full}. Certain properties of matrix calculus and perturbation theory are simplified if the perturbation changes the zeroth-order state in a linear fashion, i.e., $\nu=\epsilon \tilde{\nu}$ for small $\epsilon$. Many applications can be reduced to this special case of our work. 

We expect our results to find use in numerical modeling, where simplified expressions and the ability to avoid extra matrix diagonalizations are high priorities. It may prove useful in modeling effects of non-idealities in a quantum system, e.g., a quantum circuit with noisy constituent gates, and in computing differential effects of environmental processes. Our results may find use in computations in continuous-variable quantum information processing involving low-photon-number bosonic states, state tomography, and quantum metrology including quantum limits of sub-diffraction imaging. Through the Choi-Jamiolkowski isomorphism between channels and states, it may be possible to extend our formalism to evaluating effects of small perturbations on quantum channels, e.g., the diamond norm for perturbed quantum channels. Another related direction could be Gaussian quantum information theory \cite{Weedbrook2012}, where matrix perturbation theory could be used to find analytic lowest-order expansions for composite quantities that depend on primary matrix functions of covariance matrices of quantum states. Finally, some speculative uses of this formalism may lie in the split-step evolution of open quantum systems, and in proving important open additivity and extremality conjectures in quantum information theory, such as the entropy photon-number inequality---the quantum version of the entropy-power inequality~\cite{Guha2008-fw}.

\acknowledgements

The authors thank Marcus Carlsson and Kaushik Seshadreesan for valuable discussions. This research was supported in part by the DARPA IAMBIC Program under contract number HR00112090128, and NSF-ERC Center for Quantum Networks awarded under grant number 1941583. The views, opinions and/or findings expressed are those of the author and should not be interpreted as representing the official views or policies of the Department of Defense or the U.S. Government.

\bibliography{Quantum_Hypothesis_Testing-Sub-Diffraction_Quantum_Hypothesis_Testing}
\end{document}